\documentclass[twocolumn,showpacs,preprintnumbers,amsmath,amssymb]{revtex4}
\usepackage{mathrsfs}
\usepackage{amsfonts}
\usepackage{}
\usepackage{tikz}
\usepackage{graphicx}
\usepackage{dcolumn}
\usepackage{bm}
\usepackage{amsmath,amsthm}

\usepackage{diagbox}

\begin{document}

\preprint{}

\title{Detectors for local discrimination of sets of generalized Bell states}

\author{Cai-Hong Wang$^{1}$}%
\author{Jiang-Tao Yuan$^{1}$}%
\email[]{jtyuan@cwxu.edu.cn}
\author{Ying-Hui Yang$^{2}$}
\email[]{yangyinghui4149@163.com}
\author{Mao-Sheng Li$^{3}$}
\email[]{li.maosheng.math@gmail.com}
\author{Shao-Ming Fei$^{4}$}
\email[]{feishm@cnu.edu.cn}
\author{Zhi-Hao Ma$^{5}$}
\email[]{mazhihao@sjtu.edu.cn}
\affiliation{%
$^{1}$Department of General Education, Wuxi University, Wuxi, 214105, China\\
$^{2}$School of Mathematics and Information Science, Henan Polytechnic University, Jiaozuo, 454000, China\\
$^{3}$School  of Mathematics, South China University of Technology, GuangZhou 510640, China\\
$^{4}$School of Mathematical Sciences, Capital Normal University, Beijing 100048, China\\
$^{5}$School of Mathematical Sciences, MOE-LSC, Shanghai Jiao Tong University, Shanghai, 200240, China;\\
Shanghai Seres Information Technology Co., Ltd., Shanghai, 200040, China;
and Shenzhen Institute for Quantum Science and Engineering, Southern University of Science and Technology, Shenzhen, 518055, China}%



\begin{abstract}
A fundamental problem in quantum information processing is the discrimination among a set of orthogonal quantum states of a composite system under local operations and classical communication (LOCC). Corresponding to the LOCC indistinguishable sets of four ququad-ququad orthogonal maximally entangled states (MESs) constructed by Yu et al. [Phys. Rev. Lett. 109, 020506 (2012)], the maximum commutative sets (MCSs) were introduced as detectors for the local distinguishability of the set of generalized Bell states (GBSs), for which the detectors are sufficient to determine the LOCC distinguishability. In this work, we show how to determine all the detectors for a given GBS set.
We construct also several 4-GBS sets without detectors, most of which are one-way LOCC indistinguishable and only one is one-way LOCC distinguishable,
indicating that the detectors are not necessary for LOCC  distinguishability. Furthermore, we show that for 4-GBS sets in quantum system $\mathbb{C}^{6}\otimes\mathbb{C}^{6}$, the detectors are almost necessary for one-way LOCC distinguishability, except for one set in the sense of local unitary equivalence. The problem of one-way LOCC discrimination of 4-GBS sets in $\mathbb{C}^{6}\otimes\mathbb{C}^{6}$ is completely resolved.
\end{abstract}

\maketitle


\section{Introduction}
As a fundamental task in quantum information processing, quantum state discrimination identifies a quantum state chosen randomly from a known set of quantum states via a positive operator-valued measure (POVM). It is well known that any set of orthogonal states can be perfectly discriminated globally via a POVM [1]. However, a set of $l(\ge 3)$ orthogonal bipartite states in composite system $\mathbb{C}^{d}\otimes\mathbb{C}^{d}$ is generally indistinguishable under local operations and classical communication (LOCC) \cite{benn1999pra,walg2000prl,walg2002prl,gho2001prl}.
Such LOCC indistinguishable sets present some kind of nonlocality in the sense that
more quantum information could be inferred from global measurements than that from local operations [2].

It has been shown that any two multipartite orthogonal pure states are LOCC distinguishable \cite{walg2000prl}. A complete set of orthogonal maximally entangled states (MESs) is LOCC indistinguishable both deterministically or probabilistically, and $d+1$ or more MESs in $\mathbb{C}^{d}\otimes\mathbb{C}^{d}$ are LOCC indistinguishable \cite{horo2003prl,fan2004prl,fan2007pra,gho2004pra,nath2005jmp}, while any three generalized Bell states (GBSs) in $\mathbb{C}^{d}\otimes\mathbb{C}^{d}\ (d\geq 3)$ are always LOCC distinguishable \cite{wang2017qip,alber2000arx}. Bennett et al. \cite{benn1999pra} presented the first example of a complete set of orthogonal product states in $\mathbb{C}^{3}\otimes\mathbb{C}^{3}$ that are LOCC indistinguishable, which revealed the phenomenon of quantum nonlocality without entanglement.
Yu et al. \cite{yu2012prl} exhibited the first example of four ququad-ququad orthogonal MESs that are LOCC indistinguishable.
The local discrimination of quantum states has wide applications in data hiding, quantum secret sharing and quantum private query \cite{divi2002trans-inform,raha2015pra,yang2015sr,wei2016pra,wang2016pra,wei2018trans-comp}.

For GBS sets in $\mathbb{C}^{d}\otimes\mathbb{C}^{d}$ with cardinality $l$ ($2\leq l\leq d$) \cite{band2011njp}, many examples of LOCC indistinguishable GBS sets have been obtained
\cite{band2011njp,yu2012prl,zhang2015pra,wang2016qip,yang2018pra,yuan2019qip}. For example, Bandyopadhyay et al. \cite{band2011njp} presented the one-way LOCC indistinguishable 4-GBS sets in $\mathbb{C}^{4}\otimes\mathbb{C}^{4}$ and $\mathbb{C}^{5}\otimes\mathbb{C}^{5}$.
Some sufficient conditions on local discrimination of GBSs are also obtained
\cite{fan2004prl,tian2015pra,wang2019pra,Li20,yang2021qip}.
Fan \cite{fan2004prl} showed that any $l$ generalized Bell states (GBSs) in $\mathbb{C}^{d}\otimes\mathbb{C}^{d}$ are LOCC distinguishable
provided that $(l-1)l\leq 2d$ and $d$ is a prime number, which was extended by Tian et al. to the prime-power-dimensional quantum systems for mutually commuting qudit lattice states \cite{tian2015pra}.

In \cite{gho2004pra,band2011njp,zhang2015pra,nathan2013pra}, Bandyopadhyay et al. found the
necessary and sufficient conditions for one-way local discrimination of GBS sets, although
the conditions seem to be difficult to verify for a general GBS set. By utilizing the local unitary equivalence (LU-equivalence) classification \cite{tian2016pra,wang-yuan2021jmp} and the discriminant set, we provided necessary and sufficient conditions for local discrimination of GBS sets in $\mathbb{C}^{4}\otimes\mathbb{C}^{4}$ and $\mathbb{C}^{5}\otimes\mathbb{C}^{5}$ \cite{yuan2022quantum}, which completely solved the problem of local discrimination of GBS sets in $\mathbb{C}^{4}\otimes\mathbb{C}^{4}$. Recently, Li et al. \cite{li2022pra} introduced the maximum commutative sets (MCSs) as detectors for the local distinguishability.
They demonstrated that the detectors are more powerful than the discriminant set, that is,
if the local distinguishability of a GBS set can be determined by the discriminant set, then it can also be determined by the detectors. Therefore, detectors are crucial for solving the problem of local distinguishability. It is of significance to determine all the detectors of a given GBS set.

In this paper, we study the detectors of a given GBS set in $\mathbb{C}^{d}\otimes\mathbb{C}^{d}$. We show how to identify all the detectors capable of detecting the local distinguishability. For even $d$ $(d\geq 4)$, we construct a LOCC distinguishable 4-GBS set in $\mathbb{C}^{d}\otimes\mathbb{C}^{d}$
without detectors, indicating that the detectors are not necessary for local distinguishability. The 4-GBS set is neither an F-equivalent set \cite{hashi2021pra}, indicating that a LOCC distinguishable set may not necessarily be F-equivalent.
Finally, we demonstrate that for 4-GBS sets in $\mathbb{C}^{6}\otimes\mathbb{C}^{6}$,
the detectors are almost necessary for determining the local distinguishability.
Namely, having a detector is almost necessary and sufficient for the local distinguishability in this case.

The rest of this paper is organized as follows. In Section II, we recall some relevant notions and results. In Section III, we show how to find all detectors that determine the local distinguishability of GBS sets in $\mathbb{C}^{d}\otimes\mathbb{C}^{d}$. In Section IV, we provide several 4-GBS sets without detectors, only one of which is LOCC distinguishable,
showing that the detectors are not necessary for LOCC distinguishable sets. In Section V, for 4-GBS sets in $\mathbb{C}^{6}\otimes\mathbb{C}^{6}$, we show that having a detector is almost necessary for the GBS set to be LOCC distinguishable. We conclude in Section VI.

\section{Preliminaries}

\theoremstyle{remark}
\newtheorem{definition}{\indent Definition}
\newtheorem{lemma}{\indent Lemma}
\newtheorem{theorem}{\indent Theorem}
\newtheorem{corollary}{\indent Corollary}
\newtheorem{example}{\indent Example}
\newtheorem{proposition}{\indent Proposition}
\newtheorem{problem}{\indent Problem}
\newtheorem{conjecture}{\indent Conjecture}
\newtheorem{observation}{\indent Observation}

\def\QEDclosed{\mbox{\rule[0pt]{1.3ex}{1.3ex}}}
\def\QED{\QEDclosed}
\def\proof{\indent{\em Proof}.}
\def\endproof{\hspace*{\fill}~\QED\par\endtrivlist\unskip}

Consider a $d$-dimensional Hilbert space with computational basis $\{|j\rangle\}_{j=0}^{d-1}$.
Let $U_{m,n}=X^{m}Z^{n}$, $m,n\in\mathbb{Z}_{d}$, be the generalized Pauli matrices (GPMs)
constituting a basis of unitary operators, where $X|j\rangle=|j+1$ mod $d\rangle$, $Z|j\rangle=\omega^{j}|j\rangle$, $\omega=e^{2\pi i/d}$ and $\mathbb{Z}_{d}=\{0,1,\ldots,d-1\}$. The canonical maximally entangled state $|\Phi\rangle$ in $\mathbb{C}^{d}\otimes\mathbb{C}^{d}$ is given by $|\Phi\rangle=(1/\sqrt{d})\sum_{j=0}^{d-1}|jj\rangle$.
It is known that $(I\otimes U)|\Phi\rangle=(U^{T}\otimes I)|\Phi\rangle$, where $T$ means matrix transposition and $U$ is unitary. Any MES can be written as $|\Psi\rangle=(I\otimes U)|\Phi\rangle$. If $U=X^{m}Z^{n}$, the states
\begin{eqnarray*}
|\Phi_{m,n}\rangle=(I\otimes X^{m}Z^{n})|\Phi\rangle
\end{eqnarray*}
are called generalized Bell states (GBSs). Note that there is a one-to-one correspondence between the MESs and the unitaries, and the GBSs are given by the GPMs. It is convenient to denote a GBS set $\{(I\otimes X^{m_i}Z^{n_i})|\Phi\rangle \}$ by $\{X^{m_i}Z^{n_i} \}$ or $\{  (m_i, n_i) \}$. We denote $\mathcal{S}:=\{ |\Phi_{m_i,n_i}\rangle\}=\{  X^{m_i}Z^{n_i} \}=\{  (m_i, n_i) \}$ in the following without confusion. For a given $l$-GBS set $\mathcal{S}=\{X^{m_{i}}Z^{n_{i}}|m_{i}, n_{i}\in\mathbb{Z}_{d}, 1\leq i\leq l\}$ with $2\le l\le d$,
the difference set $\Delta \mathcal{S}$ of $\mathcal{S}$ is defined by
\begin{eqnarray*}
&\Delta \mathcal{S}=\{U_{j}U_{k}^{\dag}|U_{j}, U_{k}\in \mathcal{S}, j\neq k\}.
\end{eqnarray*}
Up to a phase, we can identify $\Delta \mathcal{S}$ as the set
\begin{eqnarray*}
\{(m_{j}-m_{k},n_{j}-n_{k})|(m_j, n_j), (m_k, n_k)\in \mathcal{S}, j\neq k\}.
\end{eqnarray*}

According to the definitions of $l$-GBS set $\mathcal{S}$ and its difference set  $\Delta\mathcal{S}$, each element $U_{j}U_{k}^{\dag}$ in $\Delta\mathcal{S}$ are determined by two different GBSs $U_{j}$ and $U_{k}$ in $\mathcal{S}$, so $\Delta\mathcal{S}$ contains at most $l(l-1)$ GPMs.
In the set $\Delta \mathcal{S}$, we use the GPM $(m_{j}-m_{k},n_{j}-n_{k})$ to refer to the element $U_{j}U_{k}^{\dag}$ uniformly, and it is easy to check that this does not change the commutation relationship between elements in $\Delta\mathcal{S}$.

\subsection{F-equivalence and discriminant set}

A set $\mathcal{S}$ is one-way LOCC distinguishable if all $m_{i}$ $(i = 1,..., l)$ are distinct \cite{fan2004prl}. Such sets are called F type \cite{hashi2021pra}. A GBS set $\mathcal{W}$ is called F-equivalent if $\mathcal{W}$ can be transformed to a F-type set under
local unitary operations.
\begin{lemma}\label{th2.1}\cite{hashi2021pra}
An $l$-GBS set $\mathcal{S}$ in $\mathbb{C}^{d}\otimes\mathbb{C}^{d}$ is F-equivalent, and thus one-way LOCC distinguishable, if and only if $m_{i}\alpha + n_{i}\beta$ $(i = 1,...,l)$ are all distinct modulo $d$ for some integers $\alpha$ and $\beta$. Especially a $d$-GBS set $\mathcal{S}$ in $\mathbb{C}^{d}\otimes\mathbb{C}^{d}$ is LOCC distinguishable if and only if $\mathcal{S}$ is F-equivalent for prime $d$.
\end{lemma}

Under certain conditions, a $d$-GBS set is F-equivalent if and only if it is LOCC distinguishable. This is an unexpected result, as local distinguishability is much more complex than F-equivalence. An F-equivalence GBS set is one-way LOCC distinguishable, but the discriminant condition in Lemma \ref{th2.1} is difficult to verify.

Two unitaries $A$ and $B$ are called Weyl commutative if
$AB=zBA$, where $z$ is a complex number \cite{petz2008book}.

\begin{lemma}\label{lem2.1}\cite{yuan2020jpa,wang-yuan2021jmp}
For two unitary matrices $A$ and $B$, if they are not commutative and satisfy the Weyl commutation relation, then each eigenvector $|u\rangle$ of $A$ satisfies $\langle u|B|u\rangle=0$.
Especially, for an arbitrary GBS set $\mathcal{S}$,
if there is a GPM $T$ which is not commutative to every GPM $U$ in $\Delta \mathcal{S}$,
then each eigenvector $|v\rangle$ of $T$ satisfies $\langle v|U|v\rangle=0$
and the set $\mathcal{S}$  is one-way LOCC distinguishable.
\end{lemma}
Now let's briefly introduce the local discrimination protocol for a given GBS set $\mathcal{S}=\{|\Phi_{m_i,n_i}\rangle\}$ based on the eigenvector $|v\rangle$ of the GPM $T$ in Lemma 2 (see [8,27] for details). Since the quantum state $|v\rangle$ satisfies $\langle v|U_{i}U_{j}^{\dag}|v\rangle=0$ for  every $U_{i}U_{j}^{\dag}\in \Delta\mathcal{S}$, the states $X^{m_{i}}Z^{n_{i}}|v\rangle$ are mutually orthogonal. Alice needs to prepare the state $|v\rangle$ on an ancillary system. Then Alice and Bob use the teleportation scheme to teleport the state $|v\rangle$ via the unknown GBS $|\Phi_{m_j,n_j}\rangle$, which is given to be identified.
Upon completion of the teleportation, the state on Bob$^{,}$s side is actually $X^{m_{j}}Z^{n_{j}}|v\rangle$.
As these $l$ numbers of possible output states are orthogonal to each other, therefore, they can be distinguished by Bob. Hence the states $\{|\Phi_{m_i,n_i}\rangle\}_{i = 1}^{l}$ can be distinguished by using one-way LOCC only.
The key quantum state $|v\rangle$ is the eigenvector of the GPM $T$ in Lemma 2. For a detector (or an MCS) $\mathcal{C}$ of $\mathcal{S}$, the key vector is taken as the common eigenvector of the detector (see equation (2)  for the detection range of a detector $\mathcal{C}$).

In order to LOCC distinguish a GBS set, according to Lemma \ref{lem2.1},
we need to find a GPM $T$ that is not commutative to every GPM $U$ in $\Delta \mathcal{S}$.
For two GPMs $(m,n)$ and $(s,t)$, they are commutative if and only if $ns-mt=0 \mod d$.
Let $m, n\in \mathbb{Z}_{d}$, and $S(m, n)$ the set of elements in $\mathbb{Z}_{d}\times \mathbb{Z}_{d}$ that commute with $(m, n)$. In \cite{yuan2022quantum} the authors defined the so-called discriminant set,
\begin{eqnarray}
\mathcal{D}(\mathcal{S})\triangleq (\mathbb{Z}_{d}\times \mathbb{Z}_{d})\setminus \bigcup_{(m,n)\in\Delta\mathcal{S}} S(m,n),
\end{eqnarray}
which satisfies the condition in Lemma \ref{lem2.1} and can be used to  locally distinguish the set $\mathcal{S}$.
The following sufficient condition on local distinguishability has been derived, which is both necessary and sufficient when $d=4$,
\begin{lemma}\cite{yuan2022quantum}\label{th2.2}
Let $\mathcal{S}=\{(m_{i},n_{i})|1\le i\le l\}$ be an $l$-GBS set in  $\mathbb{C}^{d}\otimes\mathbb{C}^{d}$ with $4\le l\le d$. The set $\mathcal{S}$ is one-way LOCC distinguishable when any of the following conditions is true.
\begin{enumerate}
\item[{\rm(1)}] The discriminant set $\mathcal{D}(\mathcal{S})$ is not empty.
\item[{\rm(2)}] The set $\Delta\mathcal{S}$ is commutative.
\item[{\rm(3)}] The dimension $d$ is a composite number, and for each $(m,n)\in\Delta\mathcal{S}$, $m$ or $n$ is invertible in $\mathbb{Z}_d$.
\end{enumerate}
\end{lemma}

According to Lemma \ref{lem2.1}, the key to one-way LOCC distinguishability is to find a quantum state $|v\rangle$ that satisfies $\langle v|U_{i}U_{j}^{\dag}|v\rangle=0$ for  every $U_{i}U_{j}^{\dag}\in \Delta\mathcal{S}$. Now we show that the condition (3) in Lemma \ref{th2.2} guarantees the existence of such quantum state:
Let $d$ be a composite number and $d=st$ be a decomposition. For $(m,n)\in\Delta\mathcal{S}$ and invertible $m$ in $\mathbb{Z}_d$, the GPM $Z^{t}$ is not commutative to $(m,n)$, then each eigenstate $|\alpha\rangle$ of $Z^{t}$ satisfies $\langle \alpha|X^{m}Z^{n}|\alpha\rangle=0$. Similarly, if $n$ is invertible in $\mathbb{Z}_d$,
then each eigenstate $|\beta\rangle$ of $X^{s}$ satisfies $\langle \beta|X^{m}Z^{n}|\beta\rangle=0$. Since $X^{s}$ and $Z^{t}$ are commutative, they have common eigenstates.
Obviously, each common eigenstate $|\gamma\rangle$ satisfies $\langle \gamma|X^{m}Z^{n}|\gamma\rangle=0$. Therefore the set $\mathcal{S}$ is one-way LOCC distinguishable.

From Lemma \ref{th2.1} and the definition of discriminant set, it can be seen that
{\it the F-equivalence of a set $\mathcal{S}$ is equivalent to that $\mathcal{D}(\mathcal{S})$ is nonempty.}
The discriminant set in Lemma 2 (1) are easy to calculate, so Lemma 2 can be used more easily than Lemma 1 to determine the local distinguishability of GBS sets.
It is known that there are two special 4-GBS sets in $\mathbb{C}^{4}\otimes\mathbb{C}^{4}$:
$\mathcal{L}_{1}=\{(0,0),(0,2),(2,0),(2,2)\}$ and $\mathcal{L}_{2}=\{(0,0),(0,1),(1,0),(3,3)\}$, which
have empty discriminant set (they are not F-equivalent and can not be locally distinguished by Lemma \ref{th2.1}) \cite[Examples 2-3]{yuan2022quantum},
and are both LOCC distinguishable by (2) and (3) in Lemma \ref{th2.2}, respectively.

\subsection{Detectors}
If there exists a GPM $(s,t)$ in $\mathbb{C}^{d}$ belonging to the discriminant set $\mathcal{D}(\mathcal{S})$, then $(s,t)$ is not commutative with every GPM in $\Delta\mathcal{S}$ and the set $\mathcal{S}$ is LOCC distinguishable. Accordingly,
$X^{s} Z^{t}$ is called a detector of the local distinguishability of GBSs \cite{li2022pra}. The ability of the detector $X^{s} Z^{t}$ is defined by the set
\begin{eqnarray*}
\mathcal{D}e(X^{s} Z^{t})\triangleq (\mathbb{Z}_{d}\times \mathbb{Z}_{d})\setminus  S(s,t).
\end{eqnarray*}
Naturally, we refer to the set $\mathcal{D}e(X^{s} Z^{t})$ as the detection range of $X^{s} Z^{t}$, that is, if $\Delta\mathcal{S}\subseteq \mathcal{D}e(X^{s} Z^{t})$, the set $\mathcal{S}$ is one-way LOCC distinguishable. Obviously, each GPM in the discriminant set $\mathcal{D}(\mathcal{S})$ is a detector of the one-way LOCC distinguishability of $\mathcal{S}$. Furthermore, as a maximally commutative set (MCS) $\mathcal{C}\triangleq\{X^{s_{i}}Z^{t_{i}}\}^{n}_{i=1}$ shares a common eigenbasis,
the set can be referred to as a stronger detector with detection range
\begin{eqnarray}
\mathcal{D}e(\mathcal{C})\triangleq \bigcup_{(s,t)\in\mathcal{C}} \mathcal{D}e(X^{s} Z^{t})
=(\mathbb{Z}_{d}\times \mathbb{Z}_{d})\setminus \mathcal{C}.
\end{eqnarray}

The definition of detector here is an improvement of that in \cite{li2022pra}.
Based on the definition of the detection range $\mathcal{D}e(\mathcal{C})$, we have
$\mathcal{D}e(\mathcal{C})=(\mathbb{Z}_{d}\times \mathbb{Z}_{d})\setminus \mathcal{C}$. Hence, the condition $\Delta\mathcal{S}\cap\mathcal{C}=\emptyset$ is equivalent to the condition $\Delta\mathcal{S}\subseteq \mathcal{D}e(\mathcal{C})$.
Generally, if the common eigenbasis of an MCS composed of unitary matrices can be used to locally distinguish a GBS set, the MCS is called a detector of the GBS set.
For convenience, one usually considers the MCS composed of GPMs. If $\Delta\mathcal{S}\subseteq \mathcal{D}e(\mathcal{C})$, the set $\mathcal{S}$ is one-way LOCC distinguishable. Therefore, we call the MCS $\mathcal{C}$ a detector of the LOCC distinguishability of the set $\mathcal{S}$, and denote $\mathcal{D}etector(\mathcal{S})$ all the detectors of  $\mathcal{S}$, that is,
\begin{eqnarray}\label{eq2.3}
\mathcal{D}etector(\mathcal{S})=\{\mathcal{C}\big|\Delta\mathcal{S}\subseteq \mathcal{D}e(\mathcal{C})\}=\{\mathcal{C}\big|\Delta\mathcal{S}\cap \mathcal{C}=\emptyset\}.
\end{eqnarray}
The Lemma \ref{th2.2} is extended to Lemma \ref{th2.3} as follows through MCSs (or detectors).

\begin{lemma}\cite{li2022pra}\label{th2.3}
Let $\mathcal{S}$ be a GBS set in $\mathbb{C}^{d}\otimes\mathbb{C}^{d}$ and $\mathcal{C}$ be an MCS of GPMs on $\mathbb{C}^{d}$.
If $\Delta\mathcal{S}\cap\mathcal{C}=\emptyset$ or $\Delta\mathcal{S}\subseteq\mathcal{C}$,
then the set $\mathcal{S}$ is one-way LOCC distinguishable.
\end{lemma}

The first and third cases of Lemma \ref{th2.2} are extended to the case $\Delta\mathcal{S}\cap\mathcal{C}=\emptyset$ of Lemma \ref{th2.3}.
The detectors in Lemma \ref{th2.3} is more powerful than the discriminant set: the nonemptyness of the discriminant set $\mathcal{D}(\mathcal{S})$ is equivalent to the existence of $X^{s}Z^{t}$ that do not commute with every element of $\Delta\mathcal{S}$, that is, the set $\Delta\mathcal{S}$ is included in the detection range $\mathcal{D}e(X^{s} Z^{t})$ of $X^{s} Z^{t}$ and $X^{s}Z^{t}$ is a detector of the one-way LOCC distinguishability of the set $\mathcal{S}$.
Meanwhile such an $X^{s}Z^{t}$ can be extended to be an MCS $\mathcal{C}$ of GPMs. As the elements in $\mathcal{C}$  all commute with $X^{s}Z^{t}$, $\Delta\mathcal{S}\cap \mathcal{C}=\emptyset$ and $\mathcal{C}$ is a stronger detector of $\mathcal{S}$.
Let us summarize the relationship among the F-equivalence, discriminant set $\mathcal{D}(\mathcal{S})$ and $\mathcal{D}etector(\mathcal{S})$:
{\it the F-equivalence of a set $\mathcal{S}$ is equivalent to the nonemptyness of $\mathcal{D}(\mathcal{S})$, each GPM in the discriminant set $\mathcal{D}(\mathcal{S})$ is a detector of the one-way LOCC distinguishability of $\mathcal{S}$, each detector $X^{s}Z^{t}$ in $\mathcal{D}(\mathcal{S})$ means that there is a stronger detector in $\mathcal{D}etector(\mathcal{S})$ that contains $X^{s}Z^{t}$.}
Figure 1 intuitively illustrates the relationship among the F-equivalence, discriminant set $\mathcal{D}(\mathcal{S})$ and $\mathcal{D}etector(\mathcal{S})$, as well as
the distinguishing ability between discriminant sets and detectors.
\begin{figure}[h]\label{fig1}
\centering
\begin{tikzpicture}

\fill[blue!10] (0,0) ellipse (4 and 2);
\draw[thick] (0,0) ellipse (4 and 2); 
\fill[blue!30] (0,0) ellipse (2.5 and 1); 
\node at (0,0.2) {\{$\mathcal{S}\big|\mathcal{D}(\mathcal{S})\neq\emptyset$\}};
\node at (0,-0.2) {(or $\mathcal{S}$ is F-equivalent)};
\draw[thick] (0,0) ellipse (2.5 and 1); 
\node at (0,-1.3) {\{$\mathcal{S}\big|\mathcal{D}etector(\mathcal{S})\neq\emptyset$\}};
\end{tikzpicture}
\caption{Schematic diagram of the relationship between F-equivalence, discriminant sets and detectors: The larger elliptical region represents the set of all GBS sets with detectors.}
\end{figure}
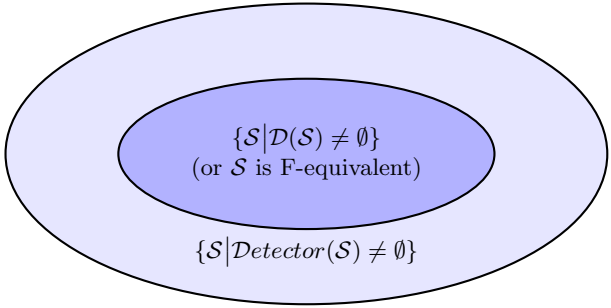

Given the importance of the detectors (or MCSs), Li et al. \cite{li2022pra} also provided a structural lemma for detectors.
\begin{lemma}[Structure characterization of MCSs]\cite{li2022pra}\label{th2.4}
Let $d\geq 3$ be an integer. For each pair $(i, j)$ in $\mathbb{Z}_{d}\times \mathbb{Z}_{d}$, where $i\neq 0$, define the following set:
\begin{equation*}    \label{eq:Cij}
\mathcal{C}_{i,j}\triangleq\{(x, y)\in \mathbb{Z}_{d}\times \mathbb{Z}_{d}\Big|
\left|\begin{array}{ccc}
i&j\\
x&y
\end{array}\right| \equiv 0\!\!\mod d, x\in i\mathbb{Z}_{d} \},
\end{equation*}
where $i\mathbb{Z}_{d}:=\{ij\in\mathbb{Z}_{d}\big|j\in\mathbb{Z}_{d}\}$.
Then $\mathcal{C}_{i,j}$ is an MCS of GBSs in $\mathbb{C}^{d}\otimes\mathbb{C}^{d}$ with exactly $d$ elements. Moreover, define $\mathcal{C}_{0,0}\triangleq \{(0, y)|y \in \mathbb{Z}_{d} \}$. Then every MCS of GBSs must be either one of $\mathcal{C}_{i,j}$ with $i \neq 0$ or $\mathcal{C}_{0,0}$. There are exactly $\sigma(d)$ (the sum of all divisors of the positive integer $d$) classes of MCSs $\mathcal{S}_{MC,d}$ which can be listed as follows:
\begin{eqnarray}
\mathcal{S}_{MC,d}\triangleq \{\mathcal{C}_{i,j}|d = ik, 0\leq j\leq k-1\}\cup\{\mathcal{C}_{0,0}\}.
\end{eqnarray}
\end{lemma}

Lemmas \ref{th2.3} and \ref{th2.4} show that the detectors are crucial for local distinguishability. We next study the constructions of the detectors for a GBS set and show how to determine all the detectors of a GBS set.

\section{Detectors of a GBS set in arbitrary dimensional systems}

Let $\mathcal{S}=\{(m_i, n_i) \}_{i=1}^{l}$ ($2\leq l\leq d$) be an $l$-GBS set in $\mathbb{C}^{d}\otimes\mathbb{C}^{d}$ ($d\geq 4$).
According to the definition of the set $\mathcal{D}etector(\mathcal{S})$ and equation (\ref{eq2.3}), an element of $\mathcal{D}etector(\mathcal{S})$ is an MCS that does not contain any element in the difference set $\Delta\mathcal{S}$. In other words, an MCS that contains any element in $\Delta\mathcal{S}$ is not an element of $\mathcal{D}etector(\mathcal{S})$. That is, $\mathcal{D}etector(\mathcal{S})$ has a useful expression based on MCSs as follows,
\begin{eqnarray}\label{eq3.1}
\mathcal{D}etector(\mathcal{S})=\mathcal{S}_{MC,d}\setminus\bigcup_{(m, n)\in\Delta\mathcal{S}}\{\mathcal{C}\big|(m, n)\in\mathcal{C}\}.
\end{eqnarray}
Hence, in order to determine $\mathcal{D}etector(\mathcal{S})$, we only need to determine the set $\{\mathcal{C}\big|(m, n)\in\mathcal{C}\}$ of each $(m, n)$ in $\Delta\mathcal{S}$.
For convenience, we denote $\{\mathcal{C}\big|(m, n)\in\mathcal{C}\}$ as $\mathcal{MCS}(m,n)$. The following will elucidate how to determine the set $\mathcal{MCS}(m,n)$ containing a GPM $(m, n)$.

According to Lemma 4, an MCS $\mathcal{C}_{i,j}$ in $\mathcal{S}_{MC,d}$ containing $(m, n)$ should satisfy the following conditions:
\begin{eqnarray*}\label{cond3.1}
i\mid d, i\mid m\ \hbox{and} \
\left|\begin{array}{ccc}
i&j\\
m&n
\end{array}\right| \equiv 0 \mod\ d.
\end{eqnarray*}
Let $d_{m}\triangleq \gcd(d,m)$ be the greatest common divisor of $d$ and $m$.
Then condition above can be rewritten as
\begin{eqnarray*}\label{}
i\mid d_{m},\ \hbox{and} \ in=jm \mod\ d.
\end{eqnarray*}
Now we need a useful result regarding homogeneous equations in one variable.
\begin{lemma}\label{lem3.1}\cite{wang2019pra,nath2000b}
Let $a,\ b,\ d$ be integers with $d \geq 1$ and $c =\gcd(a, d)$ being the greatest common divisor of $a$ and $d$. The congruence equation $ax\equiv b$ mod $d$ has a solution if and only if $c$ divides $b$. If $c$ does divide $b$ and $x_{0}$ is any solution, then the general solutions are given by $x = x_{0} + \frac{d}{c}t$, where $t \in \mathbb{Z}$. In particular,
the solution $x_{0}$ can be taken as $x_{0}=\frac{b}{c}(\frac{a}{c})^{-1}$ ($(\frac{a}{c})^{-1}$ stands for the inverse of $\frac{a}{c}$ in $\mathbb{Z}_{\frac{d}{c}}$, i.e., the integer $y$ such that $\frac{a}{c} y\equiv 1\mod  \frac{d}{c}$)
and the solutions form exactly $c$ congruence classes mod $d$.
\end{lemma}
Let $d_{m,n}\triangleq \gcd(d,m,n)$ be the greatest common divisor of $d$, $m$ and $n$.
If $m\neq 0$, due to Lemma \ref{lem3.1} and the condition $j\leq \frac{d}{i}-1$ in Lemma \ref{th2.4}, the condition above is equivalent to
\begin{eqnarray*}\label{}
i\mid d_{m},\ \big(d_{m}/i\big) \mid d_{m,n}, \\
j=\frac{in}{d_{m}}(\frac{m}{d_{m}})^{-1} (\hbox{mod}\ \frac{d}{d_{m}}) +t\frac{d}{d_{m}},t=0,\cdots,\frac{d_{m}}{i}-1.
\end{eqnarray*}
The case of $m=0$ can be considered similarly. Thus, we obtain the following theorem concerning the structure of the set $\mathcal{MCS}(m,n)$ containing $(m, n)$.

\begin{theorem}[Structure of $\mathcal{MCS}(m,n)$]\label{th3.1}
Let $(m,n)\neq(0,0)$ be a nontrivial GPM on $\mathbb{C}^{d}$ and  $\mathcal{C}_{i,j}$ be an element of $\mathcal{S}_{MC,d}$.
\begin{enumerate}
\item[{\rm(1)}] If $m\neq 0$,  $(m,n)\in \mathcal{C}_{i,j}$ if and only if $i$ and $j$ satisfy
\begin{eqnarray*}\label{}
i\mid d_{m},\ \big(d_{m}/i\big) \mid d_{m,n}, \\
j=\frac{in}{d_{m}}(\frac{m}{d_{m}})^{-1} (\hbox{mod}\ \frac{d}{d_{m}}) +t\frac{d}{d_{m}},t=0,\cdots,\frac{d_{m}}{i}-1.
\end{eqnarray*}
Specifically, when $d_{m,n}=1$, $(m,n)\in \mathcal{C}_{i,j}$ if and only if
$i=d_{m}$ and $j=n(\frac{m}{d_{m}})^{-1} (\hbox{mod}\ \frac{d}{d_{m}})$.
\item[{\rm(2)}] If $m=0$,  $(m,n)\in \mathcal{C}_{i,j}$ if and only if $i$ and $j$ satisfy $i=j=0$, or
$i\mid d,\ \big(d/i\big) \mid d_{n},\ j=0,\cdots,d/i-1.$
Specifically, when $d_{n}=1$, $(0,n)\in \mathcal{C}_{i,j}$ if and only if
$\mathcal{C}_{i,j}=\mathcal{C}_{0,0}$.
\end{enumerate}
\end{theorem}

By using Theorem \ref{th3.1} and Equation (\ref{eq3.1}), all detectors $\mathcal{D}etector(\mathcal{S})$ of $\mathcal{S}$ can be obtained for each GBS set $\mathcal{S}$.

Now, we will elaborate on how to use Theorem \ref{th3.1} to determine the detectors for the following representative 4-GBS sets in $\mathbb{C}^{4}\otimes\mathbb{C}^{4}$, thereby assessing the one-way local distinguishability of these sets.
First, it is known that all 4-GBS sets in $\mathbb{C}^{4}\otimes\mathbb{C}^{4}$ can be classified into ten equivalent
classes \cite{tian2016pra} and the representative elements of these equivalence classes are
\begin{align*}
\mathcal{K}=\{I,X^{2},Z^{2},X^{2}Z^{2}\},&\ \mathcal{L}=\{I,X,X^{2},X^{3}\},\\
\Gamma^{1}_{20}=\{I,X,Z,X^{2}\},&\ \Gamma^{1}_{31}=\{I,X,Z,X^{3}Z\},\\
\Gamma^{1}_{33}=\{I,X,Z,X^{3}Z^{3}\},&\ \Gamma^{2}_{12}=\{I,X,Z^{2},XZ^{2}\},\\
\Gamma^{2}_{30}=\{I,X,Z^{2},X^{3}\},&\ \Gamma^{1}_{12}=\{I,X,Z,XZ^{2}\},\\
\Gamma^{2}_{20}=\{I,X,Z^{2},X^{2}\},&\ \Gamma^{2}_{32}=\{I,X,Z^{2},X^{3}Z^{2}\}.
\end{align*}

According to Lemma \ref{th2.4}, there are $\sigma(4)=7$ MCSs of GPMs on $\mathbb{C}^{4}$:
$\mathcal{C}_{0,0}$, $\mathcal{C}_{1,0}$, $\mathcal{C}_{1,1}$, $\mathcal{C}_{1,2}$, $\mathcal{C}_{1,3},\mathcal{C}_{2,0}$, $\mathcal{C}_{2,1}$.
By utilizing Theorem \ref{th3.1}, for each GPM $(m,n)$, we can calculate all MCSs that contain $(m,n)$.
For instance, if $(m,n)=(0,1)$, it is known from (2) of Theorem \ref{th3.1} that only $\mathcal{C}_{0,0}$ in $\mathcal{S}_{MC,4}$ contains $(0,1)$, which implies $\mathcal{MCS}(0,1)=\{\mathcal{C}_{0,0}\}$.
Similarly, if $(m,n)=(0,2)$, since the condition ``$i\mid d,\ \big(d/i\big) \mid d_{n},\ j=0,\cdots,d/i-1$" in Theorem \ref{th3.1} (2) leads to $(i,j)=(2,0) \ \hbox{or}\ (2,1)$, it is known from (2) of Theorem \ref{th3.1} that $\mathcal{MCS}(0,2)=\{\mathcal{C}_{0,0},\mathcal{C}_{2,0},\mathcal{C}_{2,1}\}$.
Proceeding in this manner, we can derive $\mathcal{MCS}(m,n)$ for every  GPM $(m,n)$ on $\mathbb{C}^{4}$, as detailed in Table \ref{tab3.1}.

\begin{table}[h]
\caption{\label{tab3.1} $\mathcal{MCS}(m,n)$ of each GPM $(m,n)$ on $\mathbb{C}^{4}$}
\begin{tabular}{|c|c|c|c|c|}
\hline
$m$$\diagdown$ $n$&0&1&2&3\\
\hline
0&$\mathcal{S}_{MC,4}$&$\mathcal{C}_{0,0}$&$\mathcal{C}_{0,0}$,$\mathcal{C}_{2,0}$,$\mathcal{C}_{2,1}$ &$\mathcal{C}_{0,0}$  \\
\hline
1&$\mathcal{C}_{1,0}$&$\mathcal{C}_{1,1}$&$\mathcal{C}_{1,2}$ &$\mathcal{C}_{1,3}$  \\
\hline
2&$\mathcal{C}_{1,0}$,$\mathcal{C}_{1,2}$,$\mathcal{C}_{2,0}$&$\mathcal{C}_{2,1}$&$\mathcal{C}_{1,1}$,$\mathcal{C}_{1,3}$,$\mathcal{C}_{2,0}$ &$\mathcal{C}_{2,1}$  \\
\hline
3&$\mathcal{C}_{1,0}$&$\mathcal{C}_{1,3}$&$\mathcal{C}_{1,2}$ &$\mathcal{C}_{1,1}$   \\
\hline
\end{tabular}
\end{table}

By using Table \ref{tab3.1} and Equation (\ref{eq3.1}), all detectors of 4-GBS sets in $\mathbb{C}^{4}\otimes\mathbb{C}^{4}$ can be easily identified,
thereby determining their local distinguishability, see Table \ref{tab3.2} for detectors of 10 representative 4-GBS sets.
\begin{table}[ht]
\caption{$\mathcal{D}etector(\mathcal{S})$ of ten 4-GBS sets  in $\mathbb{C}^{4}\otimes\mathbb{C}^{4}$}
\label{tab3.2}
\begin{tabular}{|c|c|c|c|c|}
\hline
$\mathcal{S}$&$\Delta\mathcal{S}$&$\mathcal{D}etector(\mathcal{S})$&Y$^{\dag}$/N\\
\hline
$\mathcal{K}$&$(0,2),(2,0),(2,2)$&$\emptyset$&Y\\
\hline
$\mathcal{L}$&$(1,0),(2,0),(3,0)$&$\mathcal{C}_{0,0}$,$\mathcal{C}_{1,1}$,$\mathcal{C}_{1,3}$,$\mathcal{C}_{2,1}$&Y\\
\hline
&$(0,1),(0,3),(1,0)$&$\mathcal{C}_{1,1}$&\\
$\Gamma^{1}_{20}$&$(3,0),(2,0),(2,1)$&&Y\\
&$(2,3),(1,3),(3,1)$&&\\
\hline
$\Gamma^{1}_{31}$& $(0,1),(0,3),(1,0),(3,0)$&$\mathcal{C}_{1,1}$,$\mathcal{C}_{1,2}$,$\mathcal{C}_{2,0}$&Y\\
& $(3,1),(1,3),(2,1),(2,3)$&&\\
\hline
& $(0,1),(0,3),(1,0),(3,0)$ &$\mathcal{C}_{2,0}$ &\\
$\Gamma^{1}_{33}$&$(3,1),(1,3),(3,3),(1,1)$ & &Y\\
&$(2,3),(2,1),(3,2),(1,2)$ & &\\
\hline
$\Gamma^{2}_{12}$&$(1,0),(3,0),(1,2)$&$\mathcal{C}_{1,1}$,$\mathcal{C}_{1,3}$&Y\\
&$(3,2),(0,2)$&&\\
\hline
$\Gamma^{2}_{30}$&$(1,0),(3,0),(0,2),(1,2)$&$\mathcal{C}_{1,1}$,$\mathcal{C}_{1,3}$&Y\\
&$(3,2),(2,0),(1,3),(3,1)$&&\\
\hline
&$(1,0),(3,0),(0,3)$&$\emptyset$&\\
$\Gamma^{1}_{12}$&$(0,1),(3,2),(1,2)$&&N(\cite{tian2016pra})\\
&$(1,3),(3,1),(0,2)$&&\\
\hline
$\Gamma^{2}_{20}$&$(1,0),(3,0),(0,2),(2,0)$&$\emptyset$&N(\cite{tian2016pra})\\
&$(1,2),(3,2),(2,2)$&&\\
\hline
$\Gamma^{2}_{32}$&$ (1,0),(3,0),(0,2)$&$\emptyset$&N(\cite{tian2016pra})\\
&$(2,2),(1,2),(3,2)$&&\\
\hline
\end{tabular}\\
$^{\dag}$: $Y$ stands for LOCC distinguishable.
\end{table}
Lemmas \ref{th2.2} can be used to locally distinguish the first 7 sets except for $\Gamma^{1}_{33}$, while the discriminant set $\mathcal{D}(\Gamma^{1}_{33})$ of  $\Gamma^{1}_{33}$ is empty and it can only be distinguished by Lemma \ref{th2.3}.
This demonstrates that the capability of the detectors is indeed stronger than that of the discriminant set. Table \ref{tab3.2} illustrates that in $\mathbb{C}^{4}\otimes\mathbb{C}^{4}$, a 4-GBS set $\mathcal{S}$ is LOCC distinguishable if and only if its $\mathcal{D}etector(\mathcal{S})$ is not empty or its difference set $\Delta\mathcal{S}$ is commutative.

For the convenience of future research, we provide the set $\mathcal{MCS}(m,n)$ of each GPM $(m,n)$ on $\mathbb{C}^{5}$, $\mathbb{C}^{6}$ and $\mathbb{C}^{8}$,
see tables \ref{tab3.3}, \ref{tab3.4} and \ref{tab3.5}, respectively.
\begin{table}[htbp]
\caption{\label{tab3.3} $\mathcal{MCS}(m,n)$ of each GPM $(m,n)$ on $\mathbb{C}^5$}
\begin{tabular}{|c|c|c|c|c|c|}
\hline
$m\diagdown n$&0&1&2&3&4\\
\hline
0&$\mathcal{S}_{MC,5}$&$\mathcal{C}_{0,0}$&$\mathcal{C}_{0,0}$&$\mathcal{C}_{0,0}$&$\mathcal{C}_{0,0}$  \\
\hline
1&$\mathcal{C}_{1,0}$&$\mathcal{C}_{1,1}$&$\mathcal{C}_{1,2}$ &$\mathcal{C}_{1,3}$ &$\mathcal{C}_{1,4}$ \\
\hline
2&$\mathcal{C}_{1,0}$&$\mathcal{C}_{1,3}$&$\mathcal{C}_{1,1}$ &$\mathcal{C}_{1,4}$&$\mathcal{C}_{1,2}$  \\
\hline
3&$\mathcal{C}_{1,0}$&$\mathcal{C}_{1,2}$&$\mathcal{C}_{1,4}$ &$\mathcal{C}_{1,1}$&$\mathcal{C}_{1,3}$   \\
\hline
4&$\mathcal{C}_{1,0}$&$\mathcal{C}_{1,4}$&$\mathcal{C}_{1,3}$ &$\mathcal{C}_{1,2}$&$\mathcal{C}_{1,1}$   \\
\hline
\end{tabular}
\end{table}

\begin{table}[hbtp]
\centering
\caption{\label{tab3.4} $\mathcal{MCS}(m,n)$ of each GPM $(m,n)$ on $\mathbb{C}^6$}
\begin{tabular}{|c|c|c|c|c|c|c|}
\hline
$m\diagdown n$&0&1&2&3&4&5\\
\hline
0&$\mathcal{S}_{MC,6}$&$\mathcal{C}_{0,0}$&$\mathcal{C}_{0,0}$,$\mathcal{C}_{3,0}$
&$\mathcal{C}_{0,0}$,$\mathcal{C}_{2,0}$&$\mathcal{C}_{0,0}$,$\mathcal{C}_{3,0}$&$\mathcal{C}_{0,0}$  \\
&&&$\mathcal{C}_{3,1}$&$\mathcal{C}_{2,1}$,$\mathcal{C}_{2,2}$&$\mathcal{C}_{3,1}$&\\
\hline
1&$\mathcal{C}_{1,0}$&$\mathcal{C}_{1,1}$&$\mathcal{C}_{1,2}$ &$\mathcal{C}_{1,3}$ &$\mathcal{C}_{1,4}$ &$\mathcal{C}_{1,5}$\\
\hline
2&$\mathcal{C}_{1,0}$,$\mathcal{C}_{1,3}$&$\mathcal{C}_{2,1}$&$\mathcal{C}_{1,1}$,$\mathcal{C}_{1,4}$ &$\mathcal{C}_{2,0}$&$\mathcal{C}_{1,2}$,$\mathcal{C}_{1,5}$ &$\mathcal{C}_{2,2}$ \\
&$\mathcal{C}_{2,0}$&&$\mathcal{C}_{2,2}$&&$\mathcal{C}_{2,1}$&\\
\hline
3&$\mathcal{C}_{1,0}$,$\mathcal{C}_{1,2}$&$\mathcal{C}_{3,1}$&$\mathcal{C}_{3,0}$ &
$\mathcal{C}_{1,1}$,$\mathcal{C}_{1,3}$&$\mathcal{C}_{3,0}$  &$\mathcal{C}_{3,1}$ \\
&$\mathcal{C}_{1,4}$,$\mathcal{C}_{3,0}$&&&$\mathcal{C}_{1,5}$,$\mathcal{C}_{3,1}$&&\\
\hline
4&$\mathcal{C}_{1,0}$,$\mathcal{C}_{1,3}$&$\mathcal{C}_{2,2}$&$\mathcal{C}_{1,2}$,$\mathcal{C}_{1,5}$ &$\mathcal{C}_{2,0}$&$\mathcal{C}_{1,1}$,$\mathcal{C}_{1,4}$ &$\mathcal{C}_{2,1}$ \\
&$\mathcal{C}_{2,0}$&&$\mathcal{C}_{2,1}$&&$\mathcal{C}_{2,2}$ &\\
\hline
5&$\mathcal{C}_{1,0}$&$\mathcal{C}_{1,5}$&$\mathcal{C}_{1,4}$&$\mathcal{C}_{1,3}$ &$\mathcal{C}_{1,2}$&$\mathcal{C}_{1,1}$  \\
\hline
\end{tabular}
\end{table}

\begin{table*}[hbtp]
\centering
\caption{\label{tab3.5} $\mathcal{MCS}(m,n)$ of each GPM $(m,n)$ on $\mathbb{C}^8$}
\begin{tabular}{|c|c|c|c|c|c|c|c|c|}
\hline
$m\diagdown n$&0&1&2&3&4&5&6&7\\
\hline
0&$\mathcal{S}_{MC,8}$&$\mathcal{C}_{0,0}$&$\mathcal{C}_{0,0}$,$\mathcal{C}_{4,0}$,$\mathcal{C}_{4,1}$
&$\mathcal{C}_{0,0}$&$\mathcal{C}_{0,0}$,$\mathcal{C}_{2,0}$,$\mathcal{C}_{2,1}$,$\mathcal{C}_{2,2}$,$\mathcal{C}_{2,3}$,$\mathcal{C}_{4,0}$,$\mathcal{C}_{4,1}$
&$\mathcal{C}_{0,0}$&$\mathcal{C}_{0,0}$,$\mathcal{C}_{4,0}$,$\mathcal{C}_{4,1}$&$\mathcal{C}_{0,0}$  \\
\hline
1&$\mathcal{C}_{1,0}$&$\mathcal{C}_{1,1}$&$\mathcal{C}_{1,2}$ &$\mathcal{C}_{1,3}$ &$\mathcal{C}_{1,4}$ &$\mathcal{C}_{1,5}$&$\mathcal{C}_{1,6}$&$\mathcal{C}_{1,7}$  \\
\hline
2&$\mathcal{C}_{1,0}$,$\mathcal{C}_{1,4}$,$\mathcal{C}_{2,0}$&$\mathcal{C}_{2,1}$&$\mathcal{C}_{1,1}$,$\mathcal{C}_{1,5}$,$\mathcal{C}_{2,2}$ &$\mathcal{C}_{2,3}$&$\mathcal{C}_{1,2}$,$\mathcal{C}_{1,6}$,$\mathcal{C}_{2,0}$ &$\mathcal{C}_{2,1}$
&$\mathcal{C}_{1,3}$,$\mathcal{C}_{1,7}$,$\mathcal{C}_{2,2}$&$\mathcal{C}_{2,3}$  \\
\hline
3&$\mathcal{C}_{1,0}$&$\mathcal{C}_{1,3}$&$\mathcal{C}_{1,6}$&$\mathcal{C}_{1,1}$&$\mathcal{C}_{1,4}$ &$\mathcal{C}_{1,7}$&$\mathcal{C}_{1,2}$&$\mathcal{C}_{1,5}$  \\
\hline
4&$\mathcal{C}_{1,0}$,$\mathcal{C}_{1,2}$,$\mathcal{C}_{1,4}$,$\mathcal{C}_{1,6}$,$\mathcal{C}_{2,0}$,$\mathcal{C}_{2,2}$,$\mathcal{C}_{4,0}$&$\mathcal{C}_{4,1}$ &$\mathcal{C}_{2,1}$,$\mathcal{C}_{2,3}$,$\mathcal{C}_{4,0}$&$\mathcal{C}_{4,1}$
&$\mathcal{C}_{1,1}$,$\mathcal{C}_{1,3}$,$\mathcal{C}_{1,5}$,$\mathcal{C}_{1,7}$,$\mathcal{C}_{2,0}$,$\mathcal{C}_{2,2}$,$\mathcal{C}_{4,0}$&$\mathcal{C}_{4,1}$
&$\mathcal{C}_{2,1}$,$\mathcal{C}_{2,3}$,$\mathcal{C}_{4,0}$&$\mathcal{C}_{4,1}$  \\
\hline
5&$\mathcal{C}_{1,0}$&$\mathcal{C}_{1,5}$&$\mathcal{C}_{1,2}$&$\mathcal{C}_{1,7}$ &$\mathcal{C}_{1,4}$&$\mathcal{C}_{1,1}$&$\mathcal{C}_{1,6}$&$\mathcal{C}_{1,3}$  \\
\hline
6&$\mathcal{C}_{1,0}$,$\mathcal{C}_{1,4}$,$\mathcal{C}_{2,0}$&$\mathcal{C}_{2,1}$&$\mathcal{C}_{1,3}$,$\mathcal{C}_{1,7}$,$\mathcal{C}_{2,2}$&$\mathcal{C}_{2,3}$ &$\mathcal{C}_{1,2}$,$\mathcal{C}_{1,6}$,$\mathcal{C}_{2,0}$&$\mathcal{C}_{2,1}$&$\mathcal{C}_{1,1}$,$\mathcal{C}_{1,5}$,$\mathcal{C}_{2,2}$&$\mathcal{C}_{2,3}$  \\
\hline
7&$\mathcal{C}_{1,0}$&$\mathcal{C}_{1,7}$&$\mathcal{C}_{1,6}$&$\mathcal{C}_{1,5}$&$\mathcal{C}_{1,4}$&$\mathcal{C}_{1,3}$ &$\mathcal{C}_{1,2}$&$\mathcal{C}_{1,1}$\\
\hline
\end{tabular}
\end{table*}

\section{Some 4-GBS sets without detectors}

Theorem \ref{th3.1} provides the set $\mathcal{D}etector(\mathcal{S})$ of a GBS set $\mathcal{S}$. Specifically, we use the set $\mathcal{D}etector(\mathcal{S})$ of 4-GBS sets in $\mathbb{C}^{4}\otimes\mathbb{C}^{4}$ to easily determine the local distinguishability.
Next, we study some 4-GBS sets without detectors. These GBS sets indicate that detectors are not necessary for local distinguishability.

\subsection{A 4-GBS set in $\mathbb{C}^{d}\otimes\mathbb{C}^{d}$ with even $d$}

Inspired by Wang et al. \cite{wang2017qip} and the definition of MCS $\mathcal{C}_{ij}$ in Lemma \ref{th2.4}, we can show the following Observation 1 (see Appendix A for the complete proof).
\begin{observation}\label{ex4.1}
The 4-GBS set $\mathcal{S}_{1}=\{(0,0),(0,\frac{d}{2}),(\frac{d}{2},0),(\frac{d}{2},\frac{d}{2})\}$ in $\mathbb{C}^{d}\otimes\mathbb{C}^{d}$ is local distinguishable. In particular, for all even $d\geq 4$, the set $\mathcal{D}etector(\mathcal{S}_{1})$ is empty.
If $d=2(2\ell+1)$ for some positive integer $\ell$, then $\Delta\mathcal{S}_{1}$ is not contained in any MCS.
\end{observation}
This observation illustrates that the commutativity of the difference set is not a necessary condition for the local distinguishability of the 4-GBS set $\mathcal{S}_{1}$.
There should be conditions that are more general than the commutativity of the difference set. See Theorem \ref{th5.1} later on for details.

\subsection{Three 4-GBS sets in $\mathbb{C}^{6}\otimes\mathbb{C}^{6}$}

\begin{observation}\label{ex4.2}
The $\mathcal{D}etector(\mathcal{S}_{2})$ of the 4-GBS set $\mathcal{S}_{2}=\{(0,0),(0,1),(0,3),(3,0)\}$ in $\mathbb{C}^{6}\otimes\mathbb{C}^{6}$
is empty and $\mathcal{S}_{2}$ is one-way local indistinguishable.
\end{observation}
To prove the observation, we need the following lemma.

\begin{lemma}[\cite{gho2004pra,band2011njp,zhang2015pra}]\label{lem4.1}
An $l$-GBS set $\{|\Phi_{m_{j}n_{j}}\rangle\}_{j=1}^{l}$ in $\mathbb{C}^{d}\otimes\mathbb{C}^{d}$ is one-way LOCC distinguishable
if and only if there exists at least one state $|\alpha\rangle$ for which the set $\{U_{m_{j},n_{j}}|\alpha\rangle\}_{j=1}^{l}$ is pairwise orthogonal.
\end{lemma}

The proof of Observation 2 is by contradiction. Assume that $\mathcal{S}_{2}$ is LOCC distinguishable. According to Lemma \ref{lem4.1}, there exists one state $|\alpha\rangle$ such that the set of post-measurement states $\{U_{m_{j},n_{j}}|\alpha\rangle\}_{j=1}^{l}$ is mutually orthogonal. Thus, we obtain a set of orthogonal relations $\langle\alpha|U_{m,n}|\alpha\rangle=0$ for $(m, n)\in\Delta \mathcal{S}_{2}$, which corresponds to a homogeneous system of equations with the unit vector $|\alpha\rangle$ as a solution. According to the properties of homogeneous systems of equations, if the number of effective equations is sufficiently large, the system will have a unique solution (the trivial solution). This contradicts the existence of a non-zero solution such as $|\alpha\rangle$, leading us to conclude that set $\mathcal{S}_{2}$ is LOCC indistinguishable. For the convenience, we have provided the complete proof in Appendix B.

Similar to the proof of Observation 2, we can also obtain Observation 3 (for the complete proof, see Appendix C).
\begin{observation}\label{ex4.3}
Let $\mathcal{S}_{3}=\{(0,0),(0,2),(0,4),(2,0)\}$ be a 4-GBS set in $\mathbb{C}^{6}\otimes\mathbb{C}^{6}$. Then $\mathcal{D}etector(\mathcal{S}_{3})$ is empty and $\mathcal{S}_{3}$ is one-way local indistinguishable.
\end{observation}

We need the following Lemma \ref{lem4.2}, see the proof in Appendix D, to illustrate the next observation.

\begin{lemma}\label{lem4.2}
Let $\mathcal{S}=\{(m_{j},n_{j})\}_{j=1}^{l}$ be a set of GBSs and $I_{\frac{d}{2}}\triangleq\{(\frac{d}{2},i)\}_{i=0}^{\frac{d}{2}}$.
If there exists $i_{0}$ such that $0< i_{0}< d-1$ and $\Delta S\supseteq I_{\frac{d}{2}}\cup I_{0}$, where $I_{0}=\{(0,i_{0}+2k)\}_{k=0}^{\frac{d}{2}-1}$,
then $\mathcal{S}$ is one-way local indistinguishable.
\end{lemma}

\begin{observation}\label{ex4.4}
Let  $\mathcal{S}_{4}=\{(0,0),(0,1),(3,0),(3,3)\}$ be a 4-GBS set in  $\mathbb{C}^{6}\otimes\mathbb{C}^{6}$. Then $\mathcal{D}etector(\mathcal{S}_{4})$ is empty and $\mathcal{S}_{4}$ is one-way local indistinguishable.
\end{observation}

\begin{proof}
The difference set is given by $\Delta\mathcal{S}_{4}=\{(0,1),(0,3),(0,5),(3,0),(3,1),(3,2),(3,3),(3,4),(3,5)\}$.
From Table \ref{tab3.4}, it can be concluded that $\mathcal{D}etector(\mathcal{S}_{4})$ is empty. The one-way local indistinguishability of $\mathcal{S}_{4}$ is a special case $d=6$ of the Lemma \ref{lem4.2}.
\end{proof}

The three 4-GBS sets in this subsection illustrate that the GBS set without a detector is generally LOCC indistinguishable.

\section{Detector is almost necessary for local distinguishability of 4-GBS sets in $\mathbb{C}^{6}\otimes\mathbb{C}^{6}$}

Lemma \ref{th2.3} states that a GBS set with detectors is LOCC distinguishable,
indicating that detectors are sufficient for a GBS set to be LOCC distinguishable.
Meanwhile, Observation 1 illustrates that a GBS set without detectors may also be LOCC distinguishable, indicating that detectors are not necessary for a GBS set to be LOCC distinguishable. In this section, we demonstrate that the detectors are almost necessary and sufficient for LOCC distinguishable 4-GBS sets in $\mathbb{C}^{6}\otimes\mathbb{C}^{6}$.
\begin{theorem}\label{th5.1}
Let $\mathcal{S}$ be a  4-GBS set in $\mathbb{C}^{6}\otimes\mathbb{C}^{6}$.
Then the set $\mathcal{S}$ is one-way LOCC distinguishable if and only if
$\mathcal{D}etector(\mathcal{S})$ is not empty or the difference set $\Delta \mathcal{S}$ is $\{(0,3),(3,0),(3,3)\}$.
\end{theorem}

Theorem \ref{th5.1} completely resolves the problem of one-way local discrimination of 4-GBS sets in $\mathbb{C}^{6}\otimes\mathbb{C}^{6}$.
Figure 2 illustrates the set $\Omega$ of all one-way LOCC distinguishable 4-GBS sets $\mathcal{S}$ in Theorem \ref{th5.1}.

If $\mathcal{S}$ is a standard 4-GBS set ($\mathcal{S}$ contains $(0,0)$), the condition $\Delta \mathcal{S}=\{(0,3),(3,0),(3,3)\}$ in Theorem \ref{th5.1} implies that $\mathcal{S}=\{(0,0),(0,3),(3,0),(3,3)\}$. From the Table \ref{tab3.2}, only the set $\mathcal{K}$ has no detector and its difference set is commutative. Hence, Theorem \ref{th5.1} is also true for the 4-GBS sets in $\mathbb{C}^{4}\otimes\mathbb{C}^{4}$.
\begin{figure}[htb]\label{fig2}
\centering
\begin{tikzpicture}
\fill[blue!10](-4,-2) rectangle (4,2);\draw[thick] (-4,-2) rectangle (4,2); %
\fill[blue!30] (0,0.2) ellipse (3.5 and 1.5); 
\draw[thick] (0,0.2) ellipse (3.5 and 1.5); 
\node at (0,0.2) {$A$=\{$\mathcal{S}\big|\mathcal{D}etector(\mathcal{S})\neq\emptyset$\}}; %
\node at (0,-1.7) {$\Omega-A=\{\mathcal{S}\big|\Delta \mathcal{S}=\{(0,3),(3,0),(3,3)\}\}$};
\end{tikzpicture}
\caption{Schematic diagram of the set $\Omega$ of all one-way LOCC distinguishable 4-GBS sets $\mathcal{S}$ in Theorem \ref{th5.1}.}
\end{figure}
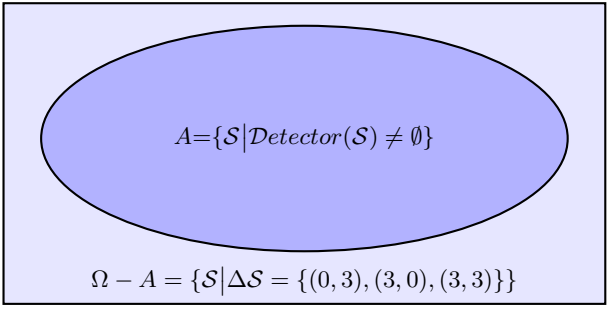

\begin{proof} If the $\mathcal{D}etector(\mathcal{S})$ is not empty,
then there is an MCS $\mathcal{C}$ in $\mathcal{D}etector(\mathcal{S})$ satisfying  $\Delta\mathcal{S}\cap\mathcal{C}=\emptyset$, and $\mathcal{S}$
is one-way LOCC distinguishable by Lemma \ref{th2.3}.
If the  difference set $\Delta \mathcal{S}$ is $\{(0,3),(3,0),(3,3)\}$,
then $\mathcal{S}$ is LU-equivalent to the set $\{(0,0),(0,3),(3,0),(3,3)\}$ and is one-way LOCC distinguishable too.

Now consider the necessity. Let $\mathcal{S}=\{(m_{i},n_{i})|1\le i\le 4\}$ be one-way LOCC distinguishable. We need to show that the $\mathcal{D}etector(\mathcal{S})$ is not empty or the  difference set $\Delta \mathcal{S}$ is $\{(0,3),(3,0),(3,3)\}$.
It is known that \cite{wang2024arx} there are 31 LU-equivalent  classes of 4-GBS sets in $\mathbb{C}^{6}\otimes \mathbb{C}^{6}$ and the representative elements of these equivalence classes are as follows (sorted in lexicographic order, see Table \ref{tab5.1}),
\begin{table}[htbp]
\centering
\caption{Representative elements of 31 LU-equivalent  classes of 4-GBS sets  in $\mathbb{C}^{6}\otimes\mathbb{C}^{6}$}
\label{tab5.1}
\begin{tabular}{|c|c|c|c|c|c|}
\hline
$\mathcal{S}$&Elements of $\mathcal{S}$&$\mathcal{S}$&Elements of $\mathcal{S}$&$\mathcal{S}$&Elements of $\mathcal{S}$\\
\hline
$\mathcal{C}^{1}$&$I,Z,Z^{2},Z^{3}$&$\mathcal{C}^{2}$&$I,Z,Z^{2},Z^{4}$&$\mathcal{C}^{3}$&$I,Z,Z^{2},X$\\
\hline
$\mathcal{C}^{4}$&$I,Z,Z^{2},X^{2}$&$\mathcal{C}^{5}$&$I,Z,Z^{2},X^{2}Z$&$\mathcal{C}^{6}$&$I,Z,Z^{2},X^{3}$\\
\hline
$\mathcal{C}^{7}$&$I,Z,Z^{2},X^{3}Z$&$\mathcal{C}^{8}$&$I,Z,Z^{3},Z^{4}$&$\mathcal{C}^{9}$&$I,Z,Z^{3},X$\\
\hline
$\mathcal{C}^{10}$&$I,Z,Z^{3},X^{2}$&$\mathcal{C}^{11}$&$I,Z,Z^{3},X^{2}Z$&$\mathcal{C}^{12}$&$I,Z,Z^{3},X^{3}$\\
\hline
$\mathcal{C}^{13}$&$I,Z,Z^{3},X^{3}Z$&$\mathcal{C}^{14}$&$I,Z,Z^{3},X^{5}$&$\mathcal{C}^{15}$&$I,Z,X,XZ$\\
\hline
$\mathcal{C}^{16}$&$I,Z,X,XZ^{2}$&$\mathcal{C}^{17}$&$I,Z,X,XZ^{3}$&$\mathcal{C}^{18}$&$I,Z,X,X^{2}Z^{2}$\\
\hline
$\mathcal{C}^{19}$&$I,Z,X,X^{3}Z^{5}$&$\mathcal{C}^{20}$&$I,Z,X,X^{4}Z^{4}$&$\mathcal{C}^{21}$&$I,Z,X^{2},X^{2}Z$\\
\hline
$\mathcal{C}^{22}$&$I,Z,X^{2},X^{2}Z^{2}$&$\mathcal{C}^{23}$&$I,Z,X^{2},X^{2}Z^{5}$&$\mathcal{C}^{24}$&$I,Z,X^{2},X^{3}Z$\\
\hline
$\mathcal{C}^{25}$&$I,Z,X^{2},X^{4}$&$\mathcal{C}^{26}$&$I,Z,X^{3},X^{3}Z$&$\mathcal{C}^{27}$&$I,Z,X^{3},X^{3}Z^{3}$\\
\hline
$\mathcal{C}^{28}$&$I,Z,X^{3},X^{3}Z^{5}$&$\mathcal{C}^{29}$&$I,Z^{2},Z^{4},X^{2}$&$\mathcal{C}^{30}$&$I,Z^{2},X^{2},X^{2}Z^{2}$\\
\hline
$\mathcal{C}^{31}$&$I,Z^{3},X^{3},X^{3}Z^{3}$&$$&$$&$$&$$\\
\hline
\end{tabular}
\end{table}
By using Table \ref{tab3.4}, all detectors of 31 sets can be identified (see Table \ref{tab5.2} for details). And then Lemma \ref{th2.3} can be used to locally distinguish the 31 sets except for the 5 sets: $\mathcal{C}^{12}$, $\mathcal{C}^{27}$ and $\mathcal{C}^{29}$-$\mathcal{C}^{31}$. From observations 1-4, we have that $\mathcal{C}^{31}$ is one-way LOCC distinguishable,
and $\mathcal{C}^{12}$, $\mathcal{C}^{27}$, $\mathcal{C}^{29}$, $\mathcal{C}^{30}$ are one-way LOCC indistinguishable. Therefore, there are 4 sets ($\mathcal{C}^{12}$, $\mathcal{C}^{27}$, $\mathcal{C}^{29}$, $\mathcal{C}^{30}$) that are LOCC indistinguishable,
while the remaining 27 sets are LOCC distinguishable.
\begin{table}[htbp]
\centering
\caption{$\mathcal{D}etector(\mathcal{S})$ of 31 GBS sets  in $\mathbb{C}^{6}\otimes\mathbb{C}^{6}$}
\label{tab5.2}
\begin{tabular}{|c|c|c|c|}
\hline
$\mathcal{S}$&$\mathcal{D}etector(\mathcal{S})$&$\mathcal{S}$&$\mathcal{D}etector(\mathcal{S})$\\
\hline
$\mathcal{C}^{1}$&$\mathcal{C}_{1,0}$,$\mathcal{C}_{1,1}$,$\mathcal{C}_{1,2}$,$\mathcal{C}_{1,3}$,$\mathcal{C}_{1,4}$,$\mathcal{C}_{1,5}$
&$\mathcal{C}^{2}$&$\mathcal{C}_{1,0}$,$\mathcal{C}_{1,1}$,$\mathcal{C}_{1,2}$,$\mathcal{C}_{1,3}$,$\mathcal{C}_{1,4}$,$\mathcal{C}_{1,5}$\\
\hline
$\mathcal{C}^{3}$&$\mathcal{C}_{1,1}$,$\mathcal{C}_{1,2}$,$\mathcal{C}_{1,4}$,$\mathcal{C}_{2,0}$,$\mathcal{C}_{2,1}$,$\mathcal{C}_{2,2}$
&$\mathcal{C}^{4}$&$\mathcal{C}_{1,1}$,$\mathcal{C}_{1,4}$\\
\hline
$\mathcal{C}^{5}$&$\mathcal{C}_{1,1}$,$\mathcal{C}_{1,2}$,$\mathcal{C}_{1,4}$,$\mathcal{C}_{1,5}$
&$\mathcal{C}^{6}$&$\mathcal{C}_{1,1}$,$\mathcal{C}_{1,3}$,$\mathcal{C}_{1,5}$,$\mathcal{C}_{2,0}$,$\mathcal{C}_{2,1}$,$\mathcal{C}_{2,2}$\\
\hline
$\mathcal{C}^{7}$&$\mathcal{C}_{1,1}$,$\mathcal{C}_{1,3}$,$\mathcal{C}_{1,5}$,$\mathcal{C}_{2,0}$,$\mathcal{C}_{2,1}$,$\mathcal{C}_{2,2}$
&$\mathcal{C}^{8}$&$\mathcal{C}_{1,0}$,$\mathcal{C}_{1,1}$,$\mathcal{C}_{1,2}$,$\mathcal{C}_{1,3}$,$\mathcal{C}_{1,4}$,$\mathcal{C}_{1,5}$ \\
\hline
$\mathcal{C}^{9}$&$\mathcal{C}_{1,1}$,$\mathcal{C}_{1,2}$,$\mathcal{C}_{1,4}$
&$\mathcal{C}^{10}$&$\mathcal{C}_{1,1}$,$\mathcal{C}_{1,2}$,$\mathcal{C}_{1,4}$,$\mathcal{C}_{1,5}$\\
\hline
$\mathcal{C}^{11}$&$\mathcal{C}_{1,1}$,$\mathcal{C}_{1,4}$
&$\mathcal{C}^{12}$&$\emptyset$\\
\hline
$\mathcal{C}^{13}$&$\mathcal{C}_{1,1}$,$\mathcal{C}_{1,3}$,$\mathcal{C}_{1,5}$
&$\mathcal{C}^{14}$&$\mathcal{C}_{1,2}$,$\mathcal{C}_{1,4}$,$\mathcal{C}_{1,5}$\\
\hline
$\mathcal{C}^{15}$&$\mathcal{C}_{1,2}$,$\mathcal{C}_{1,3}$,$\mathcal{C}_{1,4}$,$\mathcal{C}_{2,0}$,$\mathcal{C}_{2,1}$,$\mathcal{C}_{2,2}$
&$\mathcal{C}^{16}$&$\mathcal{C}_{1,3}$,$\mathcal{C}_{1,4}$,$\mathcal{C}_{2,0}$,$\mathcal{C}_{2,1}$,$\mathcal{C}_{2,2}$\\
\hline
$\mathcal{C}^{17}$&$\mathcal{C}_{1,1}$,$\mathcal{C}_{1,4}$,$\mathcal{C}_{3,0}$,$\mathcal{C}_{3,1}$
&$\mathcal{C}^{18}$&$\mathcal{C}_{1,3}$,$\mathcal{C}_{2,0}$,$\mathcal{C}_{3,0}$,$\mathcal{C}_{3,1}$\\
\hline
$\mathcal{C}^{19}$&$\mathcal{C}_{1,1}$,$\mathcal{C}_{1,2}$,$\mathcal{C}_{1,3}$,$\mathcal{C}_{1,4}$,$\mathcal{C}_{2,0}$,$\mathcal{C}_{2,1}$
&$\mathcal{C}^{20}$&$\mathcal{C}_{1,2}$,$\mathcal{C}_{1,3}$,$\mathcal{C}_{2,0}$,$\mathcal{C}_{2,1}$,$\mathcal{C}_{3,1}$\\
\hline
$\mathcal{C}^{21}$&$\mathcal{C}_{1,1}$,$\mathcal{C}_{1,2}$,$\mathcal{C}_{1,4}$,$\mathcal{C}_{1,5}$,$\mathcal{C}_{3,0}$,$\mathcal{C}_{3,1}$
&$\mathcal{C}^{22}$&$\mathcal{C}_{1,2}$,$\mathcal{C}_{1,5}$\\
\hline
$\mathcal{C}^{23}$&$\mathcal{C}_{1,1}$,$\mathcal{C}_{1,4}$,$\mathcal{C}_{2,2}$,$\mathcal{C}_{3,0}$,$\mathcal{C}_{3,1}$
&$\mathcal{C}^{24}$&$\mathcal{C}_{1,1}$,$\mathcal{C}_{1,5}$,$\mathcal{C}_{2,1}$\\
\hline
$\mathcal{C}^{25}$&$\mathcal{C}_{1,1}$,$\mathcal{C}_{1,2}$,$\mathcal{C}_{1,4}$,$\mathcal{C}_{1,5}$,$\mathcal{C}_{3,0}$,$\mathcal{C}_{3,1}$
&$\mathcal{C}^{26}$&$\mathcal{C}_{1,1}$,$\mathcal{C}_{1,3}$,$\mathcal{C}_{1,5}$,$\mathcal{C}_{2,0}$,$\mathcal{C}_{2,1}$,$\mathcal{C}_{2,2}$\\
\hline
$\mathcal{C}^{27}$&$\emptyset$
&$\mathcal{C}^{28}$&$\mathcal{C}_{1,1}$,$\mathcal{C}_{1,3}$,$\mathcal{C}_{1,5}$,$\mathcal{C}_{2,0}$,$\mathcal{C}_{2,1}$,$\mathcal{C}_{2,2}$\\
\hline
$\mathcal{C}^{29}$&$\emptyset$
&$\mathcal{C}^{30}$&$\emptyset$\\
\hline
$\mathcal{C}^{31}$&$\emptyset$&&\\
\hline
\end{tabular}
\end{table}
Recall that the local distinguishability remains invariant under local transformations.
The LOCC distinguishable set $\mathcal{S}$ is LU-equivalent to
one of the 27 one-way LOCC distinguishable sets in Table \ref{tab5.1}.
By Lemma 3 in \cite{yuan2022quantum} (or Lemma 1 in \cite{wang2024arx}),
the two conditions in Theorem \ref{th5.1} are invariant under LU-equivalence.
Then $\mathcal{D}etector(\mathcal{S})$ is either not empty or the difference set $\Delta \mathcal{S}$ is $\{(0,3),(3,0),(3,3)\}$.
\end{proof}

\section{Conclusions}
A detector of a given GBS set refers to a maximally commutative set, whose common eigenbasis can be used to locally distinguish the given set. The purpose of this paper is to illustrate how to find all the detectors of a GBS set in a bipartite system and solve the problem of one-way local discrimination based on the detectors. For a given GBS set in $\mathbb{C}^{d}\otimes\mathbb{C}^{d}$, we have used the congruence equations to illustrate how to find all the detectors (composed of GPMs) of a given set. In this way we can determine the local distinguishability for a large number of sets. Furthermore, we have presented four 4-GBS sets without detectors, of which only one is one-way LOCC distinguishable. These 4-GBS sets indicate that the detectors are not necessary for LOCC distinguishable sets,
while most sets without detectors appear to be one-way LOCC indistinguishable.
For 4-GBS sets in $\mathbb{C}^{6}\otimes\mathbb{C}^{6}$,
we have proven that a 4-GBS set is LOCC distinguishable if and only if it has a detector or its difference set is $\{(0,3),(3,0),(3,3)\}$, which indicates that a detector is almost necessary for LOCC distinguishable sets (in the sense of LU-equivalence, there is only one exception).
And then the problem of one-way local discrimination of 4-GBS sets in $\mathbb{C}^{6}\otimes\mathbb{C}^{6}$ is completely resolved.
Inspired by our conclusions, one may naturally ask if the same conclusion applies to 5-GBS sets and 6-GBS sets in $\mathbb{C}^{6}\otimes\mathbb{C}^{6}$, and if there are any other conditions or factors that affect the local distinguishability of the GBS set besides the detectors and special difference sets. Our results may highlight further investigations on such problems of local distinguishability, and give rise to intrinsic connection with the foundation of quantum mechanics and applications in quantum communications  (see \cite{bae2015jpa,barn2009aop,ber2010jmo}).

\section*{Appendix}

\subsection*{A. Proof of Observation 1}\label{appen1}
\begin{proof}
The proof here is inspired by Wang et al. \cite{wang2017qip}.
The four GBS sets in $\mathcal{S}_{1}$ are shown below,
\begin{eqnarray*}
|\Phi_{0,0}\rangle=(I\otimes I)|\Phi\rangle,\ |\Phi_{0,\frac{d}{2}}\rangle=(I\otimes Z^{\frac{d}{2}})|\Phi\rangle,\\
|\Phi_{\frac{d}{2},0}\rangle=(I\otimes X^{\frac{d}{2}})|\Phi\rangle,\ |\Phi_{\frac{d}{2},\frac{d}{2}}\rangle=(I\otimes X^{\frac{d}{2}}Z^{\frac{d}{2}})|\Phi\rangle.
\end{eqnarray*}
Alice employs the projective measurements:
$M_{k}^{\pm} = (|2k-2\rangle\pm|2k-1\rangle)(\langle2k-2|\pm\langle2k-1|)$, $k = 1, 2,..., d/2$. The post-measurement states $(M_{k}^{\pm}\otimes I)|\Phi_{m_{i},n_{i}}\rangle$ are
\begin{eqnarray*}
(|2k-2\rangle\pm|2k-1\rangle)\otimes(|2k-2\rangle\pm|2k-1\rangle),\\
(|2k-2\rangle\pm|2k-1\rangle)\otimes(|2k-2\rangle\mp|2k-1\rangle),\\
(|2k-2\rangle\pm|2k-1\rangle)\otimes(|2k-2+\frac{d}{2}\rangle\pm|2k-1+\frac{d}{2}\rangle),\\
(|2k-2\rangle\pm|2k-1\rangle)\otimes(|2k-2+\frac{d}{2}\rangle\mp|2k-1+\frac{d}{2}\rangle).
\end{eqnarray*}
Obviously, the post-measurement states of Bob$^{,}$s system are orthogonal to each other and Bob can distinguish these states exactly.

Lastly, the difference set $\Delta\mathcal{S}_{1}=\{(0,\frac{d}{2}),(\frac{d}{2},0),(\frac{d}{2},\frac{d}{2})\}$.
 We claim that for each MCS $\mathcal{C}_{ij}\in \mathcal{S}_{MC,d}$,  at least one element in $\Delta (\mathcal{S}_1)$ lies in  this MCS. For the case $i=j=0$, we have $(0, \frac{d}{2}) \in \mathcal{C}_{ij}.$ For the cases $i\neq 0$, one notes that $d=ik$ for some integer $k$.  If $i$ is odd, we have $2 \mid k$. Hence $i\mid \frac{d}{2}$ implies $\frac{d}{2} \in i \mathbb{Z}_d$. If both $i,j$ are odd,  then $$\left|\begin{array}{ccc}
i&j\\
 \frac{d}{2}& \frac{d}{2}
\end{array}\right|= (i-j) \frac{d}{2} \equiv 0 \mod\ d$$
as $(i-j)$ is even. As $\frac{d}{2} \in i \mathbb{Z}_d,$ by the definition of $\mathcal{C}_{ij} $ in Eq. \eqref{eq:Cij} we have $(\frac{d}{2}, \frac{d}{2}) \in \mathcal{C}_{ij}$.  Now if $i$ is odd  and
 $j$  is even,   we have
 $$\left|\begin{array}{ccc}
i&j\\
 \frac{d}{2}& 0
\end{array}\right|= -j \frac{d}{2} \equiv 0 \mod\ d,$$
which implies  that  $(\frac{d}{2},0) \in \mathcal{C}_{ij}$.   If $i$ is even  and
 $j$  is odd,   we have
 $$\left|\begin{array}{ccc}
i&j\\
0&  \frac{d}{2}
\end{array}\right|= i \frac{d}{2} \equiv 0 \mod\ d.
$$
Clearly, $0\in i\mathbb{Z}_d$ and we also have $(0,\frac{d}{2}) \in \mathcal{C}_{ij}$. Therefore, $\mathcal{D}etector(\mathcal{S}_{1})$ is empty. As $\frac{d}{2}= 2\ell+1$, the GPMs $(\frac{d}{2},0)$ and $(\frac{d}{2},\frac{d}{2})$ do not commute. Hence, $\Delta\mathcal{S}_{1}$ is not contained in any MCS.
\end{proof}

\subsection*{B. Proof of Observation 2}\label{appen2}
\begin{proof}
The difference set is given by $\Delta\mathcal{S}_{2}=\{(0,1),(0,2),(0,3),(0,4),(0,5),(3,0),(3,1),(3,3),(3,5)\}$.
From Table \ref{tab3.4}, it can be concluded that $\mathcal{D}etector(\mathcal{S}_{2})$ is empty.

Suppose that $\mathcal{S}_{2}$ is one-way LOCC distinguishable. Then by Lemma \ref{lem4.1}, there exists a normalized vector $|\alpha\rangle=\sum_{j=0}^{d-1}\alpha_{j}|j\rangle$ such that $\{U_{m_{j},n_{j}}|\alpha\rangle\big|(m_{j},n_{j})\in \mathcal{S}_{2}\}$ is an orthonormal set. It means that $\langle\alpha|U_{m_{j},n_{j}}^{\dag}U_{m_{k},n_{k}}|\alpha\rangle=0$, that is,
$\langle\alpha|U_{m,n}|\alpha\rangle=0$ for $(m, n)\in\Delta \mathcal{S}_{2}$.
Let $|\omega_{j}\rangle\triangleq|(1,\omega^{j},\cdots,\omega^{5j})\rangle$ and
$\{|\omega_{j}\rangle\}_{j=0}^{5}$ be a base of $\mathbb{C}^{6}$.
The condition $\Delta S\supseteq\{(0,1),(0,2),(0,3),(0,4),(0,5)\}$ implies that $\langle\alpha|Z^{i}|\alpha\rangle=0$, i.e.,
\begin{eqnarray*}\label{eq1ex4.2}
\langle(|\alpha_{j}|^{2})_{j=0}^{5}|\omega_{i}\rangle=0,\ i=1,\cdots,5.
\end{eqnarray*}
Hence, the vector $|(|\alpha_{j}|^{2})_{j=0}^{5}\rangle$ is proportional to $|(1,\cdots,1)\rangle$ and we can set $|\alpha_{j}|=k> 0$.

The condition $\Delta S\supseteq \{(3,0),(3,1),(3,3),(3,5)\}$ implies that
$0=\langle\alpha|X^{3}Z^{i}|\alpha\rangle=\Sigma_{j=0}^{5}\overline{\alpha}_{j+3}\omega^{ij}\alpha_{j}
=\langle(\overline{\alpha}_{j}\alpha_{j+3})_{j=0}^{5}|\omega_{i}\rangle$, $i=0,1,3,5$.
Let $|(\overline{\alpha}_{j}\alpha_{j+3})_{j=0}^{5}\rangle=a|\omega_{2}\rangle+b|\omega_{4}\rangle$, where $a,\ b$ are two complex numbers. Since $a|\omega_{2}\rangle+b|\omega_{4}\rangle
=|(a+b,a\omega^{2}+b\omega^{4},a\omega^{4}+b\omega^{2},a+b,a\omega^{2}+b\omega^{4},a\omega^{4}+b\omega^{2})\rangle$,
the items $\overline{\alpha}_{j}\alpha_{j+3}$ are  all real numbers.
Note that $|\alpha_{j}|=k> 0$, the items $\overline{\alpha}_{j}\alpha_{j+3}$ are either $k^{2}$ or $-k^{2}$.

In this way we obtain an equation set,
$a+b=k^{2}$ (or $-k^{2}$), $a\omega^{2}+b\omega^{4}=k^{2}$ (or $-k^{2}$) and $a\omega^{4}+b\omega^{2}=k^{2}$ (or $-k^{2}$) for $a$ and $b$. This equation set gives rise to no solutions, since the rank of the corresponding coefficient matrix is 2, which is less than the rank 3 of the augmented matrix. This is a contradiction.
\end{proof}

\subsection*{C. Proof of Observation 3}\label{appen3}
\begin{proof}
The difference set is given by $\Delta\mathcal{S}_{3}=\{(0,2),(0,4),(2,0),(2,2),(2,4),(4,0),(4,2),(4,4)\}$.
From Table \ref{tab3.4}, one concludes that $\mathcal{D}etector(\mathcal{S}_{3})$ is empty.

Suppose that $\mathcal{S}_{3}$ is one-way LOCC distinguishable.
Then by Lemma \ref{lem4.1}, there exists a normalized vector $|\alpha\rangle=\sum_{j=0}^{d-1}\alpha_{j}|j\rangle$
such that $\{U_{m_{j},n_{j}}|\alpha\rangle\big|(m_{j},n_{j})\in \mathcal{S}_{3}\}$ is an orthonormal set, which implies that $\langle\alpha|U_{m,n}|\alpha\rangle=0$ for $(m, n)\in\Delta \mathcal{S}_{2}$. Let $|\omega_{j}\rangle\triangleq|(1,\omega^{j},\cdots,\omega^{5j})\rangle$ and
$\{|\omega_{j}\rangle\}_{j=0}^{5}$ be a base of $\mathbb{C}^{6}$.
The condition $\Delta S\supseteq\{(0,2),(0,4)\}$ implies that
$0=\langle\alpha|Z^{i}|\alpha\rangle=\langle(|\alpha_{j}|^{2})_{j=0}^{d-1}|\omega_{i}\rangle$, $i=2,4.$ Since $\langle(|\alpha_{j}|^{2})_{j=0}^{d-1}|\omega_{2}\rangle
=\langle(|\alpha_{0}|^{2}+|\alpha_{3}|^{2},|\alpha_{1}|^{2}+|\alpha_{4}|^{2},|\alpha_{2}|^{2}+|\alpha_{5}|^{2})|(1,\omega^{2},\omega^{4})\rangle$
and
$\langle(|\alpha_{j}|^{2})_{j=0}^{d-1}|\omega_{4}\rangle
=\langle(|\alpha_{0}|^{2}+|\alpha_{3}|^{2},|\alpha_{1}|^{2}+|\alpha_{4}|^{2},|\alpha_{2}|^{2}+|\alpha_{5}|^{2})|(1,\omega^{4},\omega^{2})\rangle$,
it can be concluded that the vector
$|(|\alpha_{0}|^{2}+|\alpha_{3}|^{2},|\alpha_{1}|^{2}+|\alpha_{4}|^{2},|\alpha_{2}|^{2}+|\alpha_{5}|^{2})\rangle$
is proportional to $|(1,1,1)\rangle$, and we can set $|\alpha_{i}|^{2}+|\alpha_{i+3}|^{2}=k> 0$, $i=0,1,2$.

The condition $\Delta S\supseteq \{(2,0),(2,2),(2,4)\}$ implies that
$0=\langle\alpha|X^{2}Z^{i}|\alpha\rangle=\Sigma_{j=0}^{5}\overline{\alpha}_{j+2}\omega^{ij}\alpha_{j}
=\langle(\overline{\alpha}_{j}\alpha_{j+2})_{j=0}^{5}|\omega_{i}\rangle$, $i=0,2,4$.
Since $\langle(\overline{\alpha}_{j}\alpha_{j+2})_{j=0}^{5}|\omega_{i}\rangle=
\langle(\overline{\alpha}_{0}\alpha_{2}+\overline{\alpha}_{3}\alpha_{5},\overline{\alpha}_{1}\alpha_{3}+\overline{\alpha}_{4}\alpha_{0},
\overline{\alpha}_{2}\alpha_{4}
+\overline{\alpha}_{5}\alpha_{1})|(1,\omega^{2i},\omega^{4i})\rangle=0$, $i=0,2,4,$
it can be concluded that the vector
$|(\overline{\alpha}_{0}\alpha_{2}+\overline{\alpha}_{3}\alpha_{5},\overline{\alpha}_{1}\alpha_{3}+\overline{\alpha}_{4}\alpha_{0},
\overline{\alpha}_{2}\alpha_{4}+\overline{\alpha}_{5}\alpha_{1})\rangle$ is a zero vector.

Therefore, the three non-zero vectors $|(\alpha_{0},\alpha_{3})\rangle$, $|(\alpha_{2},\alpha_{5})\rangle$ and $|(\alpha_{4},\alpha_{1})\rangle$
are mutually orthogonal in $\mathbb{C}^{2}$, which is impossible and $\mathcal{S}_{3}$ is one-way local indistinguishable.
\end{proof}

\subsection*{D. Proof of Lemma \ref{lem4.2}}\label{appen4}
\begin{proof}
Suppose that $\mathcal{S}$ is one-way LOCC distinguishable.
Then there exists a normalized vector $|\alpha\rangle=\sum_{j=0}^{d-1}\alpha_{j}|j\rangle$
such that $\{U_{m_{j},n_{j}}|\alpha\rangle\big|(m_{j},n_{j})\in \mathcal{S}\}$ is a orthonormal set.
It means that $\langle\alpha|U_{m_{kj},n_{kj}}|\alpha\rangle=0$ for $(m_{jk}, n_{jk})\in\Delta \mathcal{S}$.
Hence $\Delta \mathcal{S}\supseteq I_{\frac{d}{2}}$ implies
\begin{eqnarray}\label{}\nonumber
\langle\alpha|U_{\frac{d}{2},i}|\alpha\rangle=0,~i=0,\cdots,\frac{d}{2},
\end{eqnarray}
which is equivalent to
\begin{eqnarray}\label{}\nonumber
\langle\alpha|U_{\frac{d}{2},d-i}|\alpha\rangle=0,i=0,\cdots,\frac{d}{2}.
\end{eqnarray}
Let $|\omega_{j}\rangle\triangleq|(1,\omega^{j},\cdots,\omega^{(d-1)j})\rangle$ and
$\{|\omega_{j}\rangle\}_{j=0}^{d-1}$ be a base of $\mathbb{C}^{d}$.
Then $\langle(\overline{\alpha}_{j}\alpha_{j+\frac{d}{2}})_{j=0}^{d-1}|\omega_{i}\rangle=0$ for $i\in\mathbb{Z}_{d}$.
Therefore, $|(\overline{\alpha}_{j}\alpha_{j+\frac{d}{2}})_{j=0}^{d-1}\rangle$ is a zero vector since $\{|\omega_{i}\rangle\}_{i=0}^{d-1}$ is a base.
This ensures that there is at least $\frac{d}{2}$ $\alpha_{j}=0$,
and we can assume that
\begin{eqnarray}\label{eq1lem3.0}
\alpha_{c_{1}}\neq0,\cdots,\alpha_{c_{n}}\neq0, n\leq \frac{d}{2}.
\end{eqnarray}

Denote $b_{i}\triangleq \langle \omega_{i}|(|\alpha_{j}|^{2})_{j=0}^{d-1}\rangle$
the Fourier coefficient of the vector $|(|\alpha_{j}|^{2})_{j=0}^{d-1}\rangle$ with respect to the base $\{|\omega_{i}\rangle\}_{i=0}^{d-1}$.
Let $W_{d\times d}=[\omega^{ij}]_{i,j=0}^{d-1}$ be the Vandermonde matrix generated by $(1,\omega,\cdots,\omega^{d-1})$.
Then $b_{0}=1$ ($|\alpha\rangle$ is a normalized vector) and
$|(|\alpha_{j}|^{2})_{j=0}^{d-1}\rangle=W_{d\times d}|(b_{j})_{j=0}^{d-1}\rangle$,
i.e.,
\begin{eqnarray}\label{eq2lem3.0}
W^{\dag}_{d\times d}|(|\alpha_{j}|^{2})_{j=0}^{d-1}\rangle=|(b_{j})_{j=0}^{d-1}\rangle.
\end{eqnarray}

Let $0< i_{0}< d-1$ and $\Delta S\supseteq I_{\frac{d}{2}}\cup I_{0}$,
where $I_{0}=\{(0,i_{0}+2k)\}_{k=0}^{\frac{d}{2}-1}$.
$\Delta S\supseteq I_{0}$ implies
$\langle\alpha|U_{0i}|\alpha\rangle=0$, i.e.,
\begin{eqnarray}\label{eq3lem3.0}
\langle(|\alpha_{j}|^{2})_{j=0}^{d-1}|\omega_{i}\rangle=0=b_{i},\
i=i_{0},i_{0}+2,\cdots,i_{0}+2(\frac{d}{2}-1).
\end{eqnarray}
It follows from (\ref{eq1lem3.0})-(\ref{eq3lem3.0}) that
\begin{eqnarray*}\label{eq4lem3.0}
W^{\dag}\begin{pmatrix}
i_{0}&\cdots&i_{0}+2(\frac{d}{2}-1)\\
c_{1}&\cdots&c_{n}
\end{pmatrix}
|(|\alpha_{c_{j}}|^{2})_{j=1}^{n}\rangle=|(0,\ldots,0)\rangle.
\end{eqnarray*}
This means that the homogeneous linear system
$$
W^{\dag}\begin{pmatrix}
i_{0}&\cdots&i_{0}+2(\frac{d}{2}-1)\\
c_{1}&\cdots&c_{n}
\end{pmatrix}X=|(0,\ldots,0)\rangle
$$
has a nonzero solution $|(|\alpha_{c_{j}}|^{2})_{j=1}^{n}\rangle$.
But, since $n \leq\frac{d}{2}$ and
the rank of the coefficient matrix of the homogeneous linear system is $n$, it only has zero solution. This is a contradiction.
\end{proof}

\begin{acknowledgments}
This work is supported by NSFC (Grant No. 11971151, 12371458, 12075159, 12171044, 12371132),
Wuxi University High-level Talent Research Start-up Special Fund,
Applied Basic Research Foundation(Grants No. 2023A1515012074, 2024A1515010380), Fundamental Research Funds for the Central Universities and the specific research fund of the Innovation Platform for Academicians of Hainan Province.
\end{acknowledgments}

\nocite{*}


\begin{thebibliography}{}
\bibitem{niel2004b}
M. A. Nielsen and I. L. Chuang,
Quantum Computation and Quantum Information, Cambridge University Press, Cambridge, 2004.
\bibitem{benn1999pra}
C. H. Bennett, D. P. DiVincenzo, C. A. Fuchs, T. Mor, E. Rains, P. W. Shor, J. A. Smolin, and W. K. Wootters,
Quantum nonlocality without entanglement, Phys. Rev. A \textbf{59}, 1070 (1999).
\bibitem{walg2000prl}
J. Walgate, A. J. Short, L. Hardy, and V. Vedral,
Local distinguishability of multipartite orthogonal quantum states, Phys. Rev. Lett. \textbf{85}, 4972 (2000).
\bibitem{walg2002prl}
J. Walgate and L. Hardy,
Nonlocality, asymmetry, and distinguishing bipartite states, Phys. Rev. Lett. \textbf{89}, 147901 (2002).
\bibitem{gho2001prl}
S. Ghosh, G. Kar, A. Roy, A. S. Sen (De), and U. Sen,
Distinguishability of Bell States, Phys. Rev. Lett. \textbf{87}, 277902 (2001).
\bibitem{horo2003prl}
M. Horodecki, A. Sen(De), U. Sen, and K. Horodecki,
Local indistinguishability: more nonlocality with less entanglement, Phys. Rev. Lett. \textbf{90}, 047902 (2003).
\bibitem{fan2004prl}
H. Fan,
Distinguishability and indistinguishability by local operations and classical communication,  Phys. Rev. Lett. \textbf{92}, 177905 (2004).
\bibitem{fan2007pra}
H. Fan, Distinguishing bipartite states by local operations and classical communication, Phys. Rev. A \textbf{75}, 014305 (2007).
\bibitem{gho2004pra}
S. Ghosh, G. Kar, A. Roy, and D. Sarkar,
Distinguishability of maximally entangled states, Phys. Rev. A \textbf{70}, 022304 (2004).
\bibitem{nath2005jmp}
M. Nathanson,
Distinguishing bipartite orthogonal states by LOCC: Best and worst cases, J. Math. Phys. \textbf{46}, 062103 (2005).
\bibitem{wang2017qip}
Y. L. Wang, M. S. Li, S. M. Fei and Z. J. Zheng,
The local distinguishability of any three generalized Bell states, Quant. Info. Proc. \textbf{16}, 126 (2017).
\bibitem{alber2000arx}
G. Alber, A. Delgado, N. Gisin and I. Jex,
Generalized quantum XOR-gate for quantum teleportation and state purification in arbitrary dimensional Hilbert spaces, arXiv:quant-ph/0008022
\bibitem{yu2012prl}
N. K. Yu, R. Y. Duan, and M. S. Ying,
Four locally indistinguishable ququad-ququad orthogonal maximally entangled states, Phys. Rev. Lett. \textbf{109}, 020506 (2012).
\bibitem{divi2002trans-inform}
D. P. DiVincenzo, D. W. Leung, and B. M. Terhal,
Quantum data hiding, IEEE Trans. Inf. Theory \textbf{48}, 580 (2002).
\bibitem{raha2015pra}
R. Rahaman and M. G. Parker,
Quantum scheme for secret sharing based on local distinguishability, Phys. Rev. A \textbf{91}, 022330 (2015).
\bibitem{yang2015sr}
Y. H. Yang, F. Gao, X. Wu, S. J. Qin, H. J. Zuo, and Q. Y. Wen,
Quantum secret sharing via local operations and classical communication, Sci. Rep. \textbf{5}, 16967 (2015).
\bibitem{wei2016pra}
C. Y. Wei, T. Y. Wang, and F. Gao,
Practical quantum private query with better performance in resisting joint-measurement attack, Phys. Rev. A \textbf{93}, 042318 (2016).
\bibitem{wang2016pra}
J. Wang, L. Li, H. Peng, and Y. Yang,
Quantum-secret-sharing scheme based on local distinguishability of orthogonal multiqudit entangled states, Phys. Rev. A \textbf{95}, 022320 (2017).
\bibitem{wei2018trans-comp}
C. Y. Wei, X. Q. Cai, B. Liu, T. Y. Wang, and F. Gao,
A Generic Construction of Quantum-Oblivious-Key-Transfer-Based Private Query with Ideal Database Security and Zero Failure, IEEE Trans. Comput. \textbf{67}, 2 (2018).
\bibitem{band2011njp}
S. Bandyopadhyay, S. Ghosh, and G. Kar,
LOCC distinguishability of unilaterally transformable quantum states, New J. Phys. \textbf{13}, 123013 (2011).
\bibitem{zhang2015pra}
Z. C. Zhang, K. Q. Feng, F. Gao and Q. Y. Wen,
Distinguishing maximally entangled states by one-way local operations and classical communication, Phys. Rev. A \textbf{91}, 012329 (2015).
\bibitem{wang2016qip}
Y. L. Wang, M. S. Li, Z. J. Zheng and S. M. Fei,
On small set of one-way LOCC indistinguishability of maximally entangled states, Quant. Info. Proc. \textbf{15}, 1661 (2016).
\bibitem{yang2018pra}
Y. H. Yang, J. T. Yuan, C. H. Wang and S. J. Geng,
Locally indistinguishable generalized Bell states with one-way local operations and classical communication, Phys. Rev. A \textbf{98}, 042333 (2018).
\bibitem{yuan2019qip}
J. T. Yuan, C. H. Wang, Y. H. Yang and S. J. Geng,
Constructions of one-way LOCC indistinguishable sets of generalized Bell states, Quant. Info. Proc. \textbf{18}, 145 (2019).
\bibitem{tian2015pra}
G. J. Tian, S. X. Yu, F. Gao, Q. Y. Wen and C. H. Oh,
Local discrimination of four or more maximally entangled states, Phys. Rev. A 91, 052314 (2015).
\bibitem{wang2019pra}
Y. L. Wang, M. S. Li and Z. X. Xiong,
One-way local distinguishability of generalized Bell states in arbitrary dimension, Phys. Rev. A \textbf{99}, 022307 (2019).
\bibitem{Li20}M. S. Li, S. M. Fei, Z. X. Xiong  and Y. L. Wang, Twist-teleportation-based local discrimination of maximally entangled states,
Sci. China-Phys. Mech. Astron. \textbf{63},  280312 (2020).
\bibitem{yang2021qip}
Y. H. Yang, G. F. Mu, J. T. Yuan and C. H. Wang,
Distinguishability of generalized Bell states in arbitrary dimension system via one-way local operations and classical communication,
Quant. Info. Proc. \textbf{20}, 52 (2021).
\bibitem{nathan2013pra} M. Nathanson,
Three maximally entangled states can require two-way local operations and classical communication for local discrimination,
Phys. Rev. A \textbf{88}, 062316 (2013).
\bibitem{tian2016pra}
G. J. Tian, S. X. Yu, F. Gao, Q. Y. Wen and C. H. Oh,
Classification of locally distinguishable and indistinguishable sets of maximally entangled states, Phys. Rev. A \textbf{94}, 052315 (2016).
\bibitem{wang-yuan2021jmp}
C. H. Wang, J. T. Yuan, Y. H. Yang and G. F. Mu, Local unitary classification of generalized Bell state sets in $\mathbb{C}^{5}\otimes\mathbb{C}^{5}$,
J. Math. Phys. \textbf{62}, 032203 (2021).
\bibitem{yuan2022quantum}
J. T. Yuan, Y. H. Yang and C. H. Wang, Finding out all locally indistinguishable sets of generalized Bell state,
Quantum \textbf{6}, 763 (2022).
\bibitem{li2022pra} M. S. Li, Y. L. Wang and F. Shi,
Local discrimination of generalized Bell states via commutativity, Phys. Rev. A \textbf{105}, 032455 (2022).
\bibitem{hashi2021pra} T. Hashimoto, M. Horibe, and A. Hayashi,
Simple criterion for local distinguishability of generalized Bell states in prime dimension, Phys. Rev. A \textbf{103}, 052429 (2021).
\bibitem{petz2008book} D. Petz,
Quantum Information Theory and Quantum Statistics, Theoretical and Mathematical Physics (Springer, Berlin, 2008).
\bibitem{yuan2020jpa}
J. T. Yuan, Y. H. Yang and C. H. Wang, Constructions of locally distinguishable sets of maximally entangled states which require two-way LOCC,
J. Phys. A: Math. Theor. \textbf{53}, 505304 (2020).
\bibitem{nath2000b} M. B. Nathanson,
Elementary Methods in Number Theory, Graduate Texts in Mathematics (Springer, New York, 2000).
\bibitem{wang2024arx} C. H. Wang, J. T. Yuan, M. S. Li, Y. H. Yang and S. M. Fei,
Local unitary classification of sets of generalized Bell states in $\mathbb{C}^{d}\otimes\mathbb{C}^{d}$, arXiv:2503.23073.
\bibitem{bae2015jpa} J. Bae and L. C. Kwek,
Quantum state discrimination and its applications, J. Phys. A: Math. Theor. \textbf{48}, 083001 (2015).
\bibitem{barn2009aop} S. M. Barnett and S. Croke,
Quantum state discrimination, Adv. Opt. Photonics \textbf{1}, 238-278 (2009).
\bibitem{ber2010jmo} J. A. Bergou,
Discrimination of quantum states, J. Mod. Opt. \textbf{57}, 160-180 (2010).
\end{thebibliography}

\end{document}